\documentclass[a4paper]{article}
\usepackage{CJK}
\usepackage[paperwidth=185mm,paperheight=230mm,textheight=190mm,textwidth=145mm,left=20mm,right=20mm,
top=25mm, bottom=20mm]{geometry}
\usepackage[CJKbookmarks, colorlinks,
bookmarksnumbered=true,pdfstartview=FitH,linkcolor=blue,citecolor=green]{hyperref}

\usepackage{amsmath,amssymb}
\usepackage{amsthm}
\usepackage{calc}
\usepackage{graphicx}
\usepackage{supertabular}
\usepackage{longtable}
\usepackage{float}
\usepackage{color}
\usepackage{enumerate}
\usepackage{colortbl,booktabs}
\newtheorem{theorem}{Theorem}
\newtheorem{lemma}{Lemma}

\newtheorem{remark}{Remark}
\newtheorem{assumption}{Assumption}

\usepackage{natbib}
\usepackage{bm}
\usepackage{multirow}
\usepackage{authblk}
\usepackage{enumitem}

\def\wh{\widehat}
\def\wt{\widetilde}

\newcommand{\RNum}[1]{\uppercase\expandafter{\romannumeral #1\relax}}

\begin{document}

\title{ Sparse logistic functional principal component analysis for binary data }

\author[a]{{\fontsize{12pt}{18pt}\selectfont Rou Zhong}}
\author[a]{{\fontsize{12pt}{0.5em}\selectfont Shishi Liu}}
\author[b]{{\fontsize{12pt}{0.5em}\selectfont Haocheng Li}}
\author[a]{{\fontsize{12pt}{0.5em}\selectfont Jingxiao Zhang} \thanks{zhjxiaoruc@163.com}}
\affil[a]{{\emph\fontsize{12pt}{0.5em}\selectfont Center for Applied Statistics, School of Statistics, Renmin University of China}}
\affil[b]{{\emph\fontsize{12pt}{0.5em}\selectfont Department of Mathematics and Statistics, University of Calgary}}
\date{}
\maketitle

\begin{abstract}

Functional binary datasets occur frequently in real practice, whereas discrete characteristics of the data can bring challenges to model estimation. In this paper, we propose a sparse logistic functional principal component analysis (SLFPCA) method to handle the functional binary data. The SLFPCA looks for local sparsity of the eigenfunctions to obtain convenience in interpretation. We formulate the problem through a penalized Bernoulli likelihood with both roughness penalty and sparseness penalty terms. An efficient algorithm is developed for the optimization of the penalized likelihood using majorization-minimization (MM) algorithm. The theoretical results indicate both consistency and sparsistency of the proposed method. We conduct a thorough numerical experiment to demonstrate the advantages of the SLFPCA approach. Our method is further applied to a physical activity dataset.

\textbf{Keywords}: Functional principal component analysis, penalized Bernoulli likelihood, binary data, local sparsity, MM algorithm

\end{abstract}

\section{ Introduction }

Functional principal component analysis (FPCA) is an indispensable tool in functional data analysis (FDA), for its utility in dimensionality reduction and variation mode exploration. A great many remarkable efforts have been put into FPCA, such as \citet{silverman1996smoothed}, \citet{cardot2000nonparametric}, \citet{james2000principal}, \citet{yao2005functional}, and \citet{hall2006on}, among others. In this article, we focus on functional data with binary outcomes and pursue sparsity of functional principal components (FPCs) for better interpretability.

FPCA for binary data without regard to the local sparsity has been studied by several researchers in the framework of exponential family. \citet{hall2008modelling} performed FPCA to non-Gaussian sparse longitudinal data by employing a latent Gaussian process (LGP) model with a known link function. \citet{linde2009a} considered a Bayesian FPCA approach for data from one-parameter exponential families. \citet{gertheiss2017a} conducted FPCA via a generalized additive mixed model to handle non-Gaussian cases, and established estimating procedures in both frequentist and Bayesian perspectives. \citet{li2018exponential} presented an exponential family functional principal component analysis (EFPCA) method that accommodates two-way non-Gaussian functional data. Admittedly, exploration of functional binary data is not sufficient as that of functional data from continuous distributions, for the obstacles caused by their discrete characteristics.

Nevertheless, FPCs obtained from the above methods and general FPCA approaches are almost non-zero on the whole observation interval, which increases the difficulties in interpreting the dominant variability source of the curves. Consequently, a few novel FPCA methods have been developed to gain FPCs with local sparsity, which means being strictly zero on some subintervals. \citet{chen2015localized} proposed a localized functional principal component analysis (LFPCA) method, in which they added an $L_1$ penalty on the discretized eigenfunction and constructed a deflated Fantope to estimate FPCs sequentially. An interpretable functional principal component analysis (iFPCA) method was introduced in \citet{lin2016interpretable}. They utilized an $L_0$ penalty and devised a greedy backward elimination algorithm to achieve approximate optimization. \citet{li2016supervised} made use of additional variables to incorporate some supervision information in the sparse functional principal component framework. \citet{nie2020sparse} established sparse FPCA methods taking advantage of means in functional regression. Additionally, \citet{wang2020hierarchical} and \citet{zhang2019interpretable} considered sparse FPCA in more complicated multivariate functional settings. The aforementioned techniques are only suitable for functional data from continuous distributions, thus cannot be used for functional binary data that abound in practice. Moreover, to the best of our knowledge, there has been no relevant work considering sparse FPCA for binary data in the literatures.

In this article, we propose a new sparse FPCA approach, called sparse logistic functional principal component analysis (SLFPCA), which can be applied to functional binary data under both dense and sparse designs. Specifically, a likelihood-based method is established for the proposed SLFPCA, inspired by sparse principal component analysis techniques for multivariate binary data \citep{lee2010sparse, lee2013a}.
Further, \citet{james2000principal} and \citet{zhou2008joint} executed FPCA by constructing an appropriate likelihood function for Gaussian data. Different from them, we introduce a penalized Bernoulli likelihood.
To meet the need for both optimal performance and high interpretability, two types of penalty, roughness penalty and sparseness penalty, are imposed correspondingly. The roughness penalty is commonly used in FDA to control the degree of smoothing for model fitting and obtain better estimation. And the sparseness penalty can be exerted to identify non-zero subintervals and contribute to more comprehensible and interpretable conclusions. However, it poses great challenges to optimize the penalized Bernoulli likelihood with roughness and sparseness penalties, for the reasons that the objective function is no more quadratic and the selection of tuning parameters must be taken into account. Here we employ the majorization-minimization (MM) algorithm, in which we define a simpler surrogate objective function iteratively. Moreover, Bayesian information criterion (BIC) \citep{schwarz1978estimating} is embedded in the computation for selecting tuning parameters.

Compared with the existing works, our contributions are three-fold. First, it is the first attempt to ponder sparse FPCA for binary data and the proposed SLFPCA is formulated as the optimization of a penalized Bernoulli likelihood. Second, we provide an innovative algorithm for the model fitting, which gives out satisfying estimating results. For the implementation, we also develop an R package \emph{SLFPCA}, which is available on \url{https://CRAN.R-project.org/package=SLFPCA}. Third, asymptotic properties on both consistency and sparsistency are established.

The paper is laid out as follows. In Section \ref{SecMethod}, we introduce the methodology, including the construction of penalized likelihood and computational details of the algorithm. Theoretical results are provided in Section \ref{SecTheory}. A numerical study is executed in Section \ref{SecSim} to assess the performance of our method. Section \ref{SecReal} shows a real data example on physical activity. We conclude this paper with some discussions in Section \ref{SecCon}.

\section{ Methodology }\label{SecMethod}

\subsection{Penalized Likelihood}

Consider a random process $Y(t), t \in \mathcal{T}$ with binary outcomes, where $\mathcal{T} = [0, T]$ is a bounded and closed interval. For a given time point $t$, assume that $Y(t)$ follows $\emph{Bernoulli(1, p(t))}$, where $p(t) = \Pr\{Y(t) = 1\}$. Let $X(t)$ denote the canonical parameter for $\emph{Bernoulli(1, p(t))}$, and $\{X(t), t \in \mathcal{T}\}$ is supposed to be a latent square integrable process with mean function $EX(t) = \mu(t)$ and covariance function $\Sigma(s, t) = \mbox{cov} \{X(s), X(t)\}$. Moreover, we have
\begin{align}
p(t) = \frac{\exp\{X(t)\}}{\exp\{X(t)\} + 1} \triangleq \pi\{X(t)\}. \nonumber
\end{align}
In practice, let $\{t_{ij}\mbox{:} \ i = 1,\ldots, n, j = 1, \ldots, m_i\}$ be the observation time points for $n$ independent subjects, each with $m_i$ measurements, and $\{y_{ij}\mbox{:} \ i = 1,\ldots, n, j = 1, \ldots, m_i\}$ be the corresponding observations. We further define $Y_{ij} = Y_i(t_{ij})$, where $Y_i(\cdot)$ is the random trajectory of the $i$-th subject, then $y_{ij}$ can be seen as a realization of the random variable $Y_{ij}$, thus
\begin{align}\label{PrYij}
\Pr(Y_{ij} = y_{ij}) = \pi(X_{ij})^{y_{ij}} \{1 - \pi(X_{ij})\}^{1 - y_{ij}} = \pi(q_{ij} X_{ij}),
\end{align}
where $X_{ij} = X_i(t_{ij})$ and $q_{ij} = 2 y_{ij} - 1$.

The latent process $X_i(t)$ admits the following Karhunen-Lo\`{e}ve expansion
\begin{align}\label{KL}
X_i(t) = \mu(t) + \sum_{k = 1}^{\infty} \xi_{ik} \phi_k(t),
\end{align}
where $\phi_k(t)$ is the $k$-th eigenfunction of $\Sigma(s, t)$ such that $\int_{\mathcal{T}} \phi_k^2(t) dt = 1$ while $\int_{\mathcal{T}} \phi_k(t) \phi_l(t) dt = 0$ for $l \neq k$, and $\xi_{ik}$ is the corresponding FPC score. In addition, let $\lambda_k$ be the $k$-th eigenvalue of $\Sigma(s, t)$, then $\xi_{ik}$, for $k \geq 1$, are uncorrelated random variables with $E \xi_{ik} = 0$ and $\mbox{var} (\xi_{ik}) = \lambda_k$. Considering the feasibility for practical estimation \citep{gervini2008robust, huang2014joint}, we adopt a reduced rank model
\begin{align}\label{rrmod}
X_i(t) = \mu(t) + \sum_{j = 1}^{p} \xi_{ik} \phi_k(t),
\end{align}
where $p$ is chosen in advance. In order to construct an appropriate likelihood function when the curves are not fully observed, we make use of the B-spline basis here. The reasons for selecting B-spline basis rather than other basis functions are discussed in Section \ref{SecSparPen}. Let $\{B_l(t), l = 1, \ldots, L\}$ be the B-spline basis functions on $\mathcal{T}$ with degree $d$ and knots $0 = \tau_0 < \tau_1 < \cdots < \tau_K < \tau_{K + 1} = T$, where $K$ is the number of interior knots, then we have $L = K + d + 1$. Let $\textbf{B}(t) = \{B_1(t), \ldots, B_L(t)\}^\top$. Therefore, $\mu(t)$ and $\phi_k(t)$ can be expressed as
\begin{align}
\mu(t) &= \textbf{B}(t)^\top \bm{\mu}, \nonumber \\
\phi_k(t) &= \textbf{B}(t)^\top \bm{\theta}_k, k = 1, \ldots, p, \nonumber
\end{align}
where $\bm{\mu}$ and $\bm{\theta}_k$ are the coefficients of the mean function and the $k$-th eigenfunction respectively. Let $\bm{\Theta}_{p \times L} = (\bm{\theta}_1, \cdots, \bm{\theta}_p)^\top$, then in term of (\ref{rrmod}),
\begin{align}
X_{ij} = X_i(t_{ij}) &= \textbf{B}(t_{ij})^\top \bm{\mu} + \textbf{B}(t_{ij})^\top \bm{\Theta}^\top \bm{\xi}_i \nonumber \\
&= \textbf{B}_{ij}^\top \bm{\mu} + \textbf{B}_{ij}^\top \bm{\Theta}^\top \bm{\xi}_i, \label{XBrep}
\end{align}
where $\textbf{B}_{ij} = \textbf{B}(t_{ij})$ and $\bm{\xi}_i = (\xi_{i1}, \ldots, \xi_{ip})^\top$. Combining (\ref{PrYij}) and (\ref{XBrep}), we obtain the log-likelihood function
\begin{align}
l(\bm{\mu}, \bm{\Theta}, \bm{\xi}) &= \sum_{i = 1}^n \sum_{j = 1}^{m_i} \log \pi (q_{ij} X_{ij}) \nonumber\\
&= \sum_{i = 1}^n \sum_{j = 1}^{m_i} \log \pi \{q_{ij} (\textbf{B}_{ij}^\top \bm{\mu} + \textbf{B}_{ij}^\top \bm{\Theta}^\top \bm{\xi}_i) \}, \label{loglik}
\end{align}
where $\bm{\xi}_{n \times p} = (\bm{\xi}_1, \ldots, \bm{\xi}_n)^\top$.

We next impose two types of structural regularization on the estimation in (\ref{loglik}). First, to alleviate the excessive variability of the estimated mean function and eigenfunctions, roughness penalty is exerted on $\mu(t)$ and $\phi_k(t)$ to control the degree of smoothing. We adopt the most common roughness penalties, $\int_{\mathcal{T}} \{ \mu^{(2)}(t) \}^2 dt$ and $\int_{\mathcal{T}} \{ \phi_k^{(2)}(t) \}^2 dt$, where $\mu^{(2)}(t)$ and $\phi_k^{(2)}(t)$ are the second derivatives of $\mu(t)$ and $\phi_k(t)$ respectively. Using B-spline basis, the roughness penalties are represented as
\begin{align}
\int_{\mathcal{T}} \{ \mu^{(2)}(t) \}^2 dt &= \bm{\mu}^\top \int_{\mathcal{T}} \textbf{B}^{(2)}(t) \textbf{B}^{(2)}(t)^\top dt \bm{\mu} = \bm{\mu}^\top V \bm{\mu}, \nonumber \\
\int_{\mathcal{T}} \{ \phi_k^{(2)}(t) \}^2 dt &= \bm{\theta}_k^\top \int_{\mathcal{T}} \textbf{B}^{(2)}(t) \textbf{B}^{(2)}(t)^\top dt \bm{\theta}_k = \bm{\theta}_k^\top V \bm{\theta}_k, \nonumber
\end{align}
where $\textbf{B}^{(2)}(t)$ is the second derivative of $\textbf{B}(t)$ and $V = \int_{\mathcal{T}} \textbf{B}^{(2)}(t) \textbf{B}^{(2)}(t)^\top dt$. Second, to enhance interpretability, we pursue eigenfunction estimates that reflet local sparsity. Hence, a sparseness penalty is also exerted on the eigenfunctions. The formulation of sparseness penalty is discussed at length in Section \ref{SecSparPen} and we denote it by $\mbox{PEN}_{\lambda}(\bm{\Theta})$ at present, where $\lambda$ is the tuning parameter that controls the level of sparseness. At last, the penalized likelihood method minimizes the following objective
\begin{align}\label{penloglik1}
-\sum_{i = 1}^n \sum_{j = 1}^{m_i} \log \pi \{q_{ij} (\textbf{B}_{ij}^\top \bm{\mu} &+ \textbf{B}_{ij}^\top \bm{\Theta}^\top \bm{\xi}_i) \} \nonumber \\
&+ N \kappa_{\bm{\mu}} \bm{\mu}^\top V \bm{\mu} + N \kappa_{\bm{\theta}} \sum_{k = 1}^p \bm{\theta}_k^\top V \bm{\theta}_k + N \mbox{PEN}_{\lambda}(\bm{\Theta}),
\end{align}
with respect to $\bm{\mu}, \bm{\Theta}$ and $\bm{\xi}$, where $N = \sum_{i = 1}^n m_i$, $\kappa_{\bm{\mu}}$ and $\kappa_{\bm{\theta}}$ are two tuning parameters. Note that for simplicity, we take same tuning parameters $\kappa_{\bm{\theta}}$ and $\lambda$ for all eigenfunctions.

\subsection{Sparseness Penalty}\label{SecSparPen}

We expect the estimated eigenfunctions to possess some local sparse features through the imposed sparseness penalty.
We generalize the functional SCAD penalty suggested in \citet{lin2017locally} to our FPCA framework. In specific,
\begin{align}\label{SparPen}
\mbox{PEN}_{\lambda}(\bm{\Theta}) &= \frac{K + 1}{8T} \sum_{k = 1}^{p} \int_{\mathcal{T}} p_{\lambda} (|\phi_k(t)|) dt \nonumber \\
&\approx \frac{1}{8} \sum_{k = 1}^{p} \sum_{m = 1}^{K + 1} p_{\lambda} \Bigg (\sqrt{\frac{K + 1}{T} \int_{\tau_{m - 1}}^{\tau_{m}} \phi_k^2(t) dt} \Bigg ),
\end{align}
where $p_{\lambda}(\cdot)$ is the SCAD function proposed in \citet{fan2001variable}, which is defined as
\begin{align}
p_{\lambda}(v) = \left\{ \begin{array}{rcl}
\lambda v &\mbox{when} &0 \leq v \leq \lambda \\
-\frac{v^2 - 2 a \lambda v + \lambda^2}{2(a - 1)} &\mbox{when} &\lambda < v < a \lambda \\
\frac{(a + 1) \lambda^2}{2} &\mbox{when} &v \geq a \lambda
\end{array} \right., \nonumber
\end{align}
where $a$ is chosen to be $3.7$ suggested by \citet{fan2001variable}. The local quadratic approximation is applied to (\ref{SparPen}). Specifically, for a given $v_0$ close to $v$, the local quadratic approximation can be expressed as $p_{\lambda}(|v|) \approx p_{\lambda}(|v_0|) + \frac{1}{2} \{p_{\lambda}^{'}(|v_0|)/|v_0|\}(v^2 - v_0^2)$. Further, substituting $\phi_k(t)$ with its basis representation, we finally have
\begin{align}
\mbox{PEN}_{\lambda}(\bm{\Theta}) \approx \frac{1}{8} \sum_{k = 1}^p \bm{\theta}_k^\top W_{\lambda, k} \bm{\theta}_k, \nonumber
\end{align}
with the constant term ignored, where
\begin{align}
W_{\lambda, k} = \frac{1}{2} \sum_{m = 1}^{K + 1} \Bigg \{ \frac{p_{\lambda}^{'} \big (\sqrt{\frac{K + 1}{T} \int_{\tau_{m - 1}}^{\tau_{m}} \phi_{k,0}^{2}(t) dt} \big )}{\sqrt{\frac{T}{K + 1} \int_{\tau_{m - 1}}^{\tau_{m}} \phi_{k, 0}^2(t) dt}} V_m \Bigg \}, \nonumber
\end{align}
with $V_m = \int_{\tau_{m - 1}}^{\tau_{m}} \textbf{B}(t) \textbf{B}(t)^\top dt$ and $\phi_{k, 0}(t)$ close to $\phi_k(t)$. In the iterative procedure, $\phi_{k, 0}(t)$ is replaced by the initial values or the estimates obtained from previous iteration. Hence, we aim to minimize
\begin{align}\label{penloglik2}
-\sum_{i = 1}^n \sum_{j = 1}^{m_i} \log \pi \{q_{ij} (&\textbf{B}_{ij}^\top \bm{\mu} + \textbf{B}_{ij}^\top \bm{\Theta}^\top \bm{\xi}_i) \} \nonumber \\
&+ N \kappa_{\bm{\mu}} \bm{\mu}^\top V \bm{\mu} + N \kappa_{\bm{\theta}} \sum_{k = 1}^p \bm{\theta}_k^\top V \bm{\theta}_k + \frac{N}{8} \sum_{k = 1}^p \bm{\theta}_k^\top W_{\lambda, k} \bm{\theta}_k
\end{align}
in the computation. More details are provided in Section \ref{SecAlgor}.

We complete this subsection with a discussion on the reasons for the choice of B-spline basis. From the sparseness penalty (\ref{SparPen}), we constraint the magnitude of eigenfunctions via each subinterval, that means exerting localized regularization to capture particular local features. Through basis representations, we transfer the penalization to the basis coefficients. For a general basis system, a set of sparse coefficients does not necessarily generate a function with local sparse feature, which may make trouble in the computation. On the contrary, B-spline basis enjoys the compact support property \citep{ramsay2005functional}, which elucidates that the basis is non-zero over no more than $d + 1$ adjacent subintervals. Therefore, consecutive $d + 1$ zero-valued basis coefficients indicate the resulting function being zero-valued on certain interval. This outstanding property of B-spline basis makes it crucial for the work on local sparsity, see \citet{zhou2013functional}, \citet{wang2015functional}, \citet{lin2017locally}, and \citet{tu2020estimation}.

\subsection{Algorithm}\label{SecAlgor}

The minimization of (\ref{penloglik1}) or (\ref{penloglik2}) is a tough task for the complicated expression of their first term. Thus, we first apply the MM algorithm to obtain a sequence of surrogate objective functions, which are simple enough for computation. For function $\pi(v)$, we have
\begin{align}
-\log \pi (v) \leq -\log \pi (v_0) + \frac{1}{8} [v - v_0 - 4 \{ 1 - \pi (v_{0})\}]^2, \nonumber
\end{align}
for any $v_0$ \citep{lee2010sparse}. Then the upper bound of $-\log \pi (q_{ij} X_{ij})$ can be achieved by
\begin{align}
-\log \pi (q_{ij} X_{ij}) \leq -\log \pi (q_{ij} X_{ij, 0}) + \frac{1}{8} (X_{ij} - z_{ij, 0})^2, \nonumber
\end{align}
where $X_{ij, 0}$ can be the initial value or be obtained from the last iteration, and $z_{ij, 0} = X_{ij, 0} + 4 q_{ij} \{ 1 - \pi(q_{ij} X_{ij, 0}) \}$. As the constant $-\log \pi (q_{ij} X_{ij, 0})$ has no effect on the optimization, the surrogate objective function can be written as
\begin{align}\label{surrloglik}
\sum_{i = 1}^n \sum_{j = 1}^{m_i} \{z_{ij, 0} - (\textbf{B}_{ij}^\top \bm{\mu} &+ \textbf{B}_{ij}^\top \bm{\Theta}^\top \bm{\xi}_i)\}^2 \nonumber \\
&+ N \kappa_{\bm{\mu}} \bm{\mu}^\top V \bm{\mu} + N \kappa_{\bm{\theta}} \sum_{k = 1}^p \bm{\theta}_k^\top V \bm{\theta}_k + N \sum_{k = 1}^p \bm{\theta}_k^\top W_{\lambda, k} \bm{\theta}_k,
\end{align}
where the extra multiplier $8$ can be absorbed into the tuning parameters $\kappa_{\bm{\mu}}$ and $\kappa_{\bm{\theta}}$ in the second and third terms respectively.

To optimize (\ref{surrloglik}), we consider the minimizations with respect to $\bm{\mu}, \bm{\xi}$ and $\bm{\Theta}$ sequentially. First, for fixed $\bm{\xi}$ and $\bm{\Theta}$, let $\wt{z}_{ij} = z_{ij, 0} - \textbf{B}_{ij}^\top \bm{\Theta}^\top \bm{\xi}_i$. Then we have
\begin{align}
\wh{\bm{\mu}} &= \mathop{\arg\min}\limits_{\bm{\mu}} \sum_{i = 1}^n \sum_{j = 1}^{m_i} \{\wt{z}_{ij} - \textbf{B}_{ij}^\top \bm{\mu} \}^2 + N \kappa_{\bm{\mu}} \bm{\mu}^\top V \bm{\mu} \nonumber \\
&= (\textbf{B} \textbf{B}^\top + N \kappa_{\bm{\mu}}V)^{-1} \textbf{B}^\top \wt{\textbf{Z}}, \label{muest}
\end{align}
where $\textbf{B} = (\textbf{B}_{11} \cdots \textbf{B}_{1m_1} \cdots \textbf{B}_{n m_n})^\top$ and $\wt{\textbf{Z}} = (\wt{z}_{11} \cdots \wt{z}_{1m_1} \cdots \wt{z}_{nm_n})^\top$. Second, we estimate $\bm{\xi}_k$ and $\bm{\theta}_k$ iteratively. Specifically, given $\bm{\xi}_l$ and $\bm{\theta}_l$ for $l \neq k$, define $\bar{z}_{ij} = z_{ij, 0} - \textbf{B}_{ij}^\top \wh{\bm{\mu}} - \textbf{B}_{ij}^\top \sum_{l \neq k} \xi_{ij} \bm{\theta}_l$. Subsequently, $\wh{\xi}_{ik}$ also has an explicit expression
\begin{align}
\wh{\xi}_{ik} &= \mathop{\arg\min}\limits_{\xi_{ik}} \sum_{i = 1}^n \sum_{j = 1}^{m_i} \{ \bar{z}_{ij} - \textbf{B}_{ij}^\top \bm{\theta}_k \xi_{ik} \}^2 \nonumber \\
&= \mathop{\arg\min}\limits_{\xi_{ik}} \sum_{j = 1}^{m_i} \{ \bar{z}_{ij} - \textbf{B}_{ij}^\top \bm{\theta}_k \xi_{ik} \}^2 \nonumber \\
&= \frac{\sum_{j = 1}^{m_i} \textbf{B}_{ij}^\top \bm{\theta}_k \bar{z}_{ij}}{\sum_{j = 1}^{m_i} (\textbf{B}_{ij}^\top \bm{\theta}_k)^2}. \label{xiest}
\end{align}
Then $\wh{\bm{\xi}}_k = (\wh{\xi}_{1k}, \ldots, \wh{\xi}_{nk})^\top$. On the other hand, the estimation of $\bm{\theta}_k$ is more complex for it involves the sparseness penalty and we construct a sub-iteration procedure for $\bm{\theta}_k$. The corresponding objective function can be written as
\begin{align}
\sum_{i = 1}^n \sum_{j = 1}^{m_i} \{ \bar{z}_{ij} - \xi_{ik} \textbf{B}_{ij}^\top \bm{\theta}_k \}^2 + N \kappa_{\bm{\theta}} \bm{\theta}_k^\top V \bm{\theta}_k + N \bm{\theta}_k^\top W_{\lambda, k} \bm{\theta}_k. \nonumber
\end{align}
If $W_{\lambda, k}$ is known, we have
\begin{align}
\wh{\bm{\theta}}_k = (U^\top U + N \kappa_{\bm{\theta}} V + N W_{\lambda, k})^{-1} U^\top \bar{\textbf{Z}}, \label{thetaest}
\end{align}
where $U = (\xi_{1k} \textbf{B}_{11} \cdots \xi_{1k} \textbf{B}_{1m_1} \cdots \xi_{nk} \textbf{B}_{nm_n})^\top$ and $\bar{\textbf{Z}} = (\bar{z}_{11} \cdots \bar{z}_{1m_1} \cdots \bar{z}_{nm_n})^\top$. As $W_{\lambda, k}$ depends on the value of $\bm{\theta}_k$, we update $W_{\lambda, k}$ using the new estimated $\bm{\theta}_k$ until convergence.
The appearance of small elements in $\wh{\bm{\theta}}_k$ may make $U^\top U + N \kappa_{\bm{\theta}} V + N W_{\lambda, k}$ almost singular during the sub-iteration procedure. To avoid that, we shrink the small elements to zero directly.
Furthermore, we enforce the first and last elements in $\wh{\bm{\theta}}_k$ to zero at the beginning of the sub-iteration procedure to alleviate boundary effect for the estimation of eigenfunctions.

For the sake of clarity, we summarize the algorithm as follows:
\begin{itemize}[leftmargin = 35pt]
\item[Step 1:] Give the initial value of $\bm{\mu}, \bm{\xi}$ and $\bm{\Theta}$.
\item[Step 2:] Estimate $\bm{\mu}$ using (\ref{muest}), then $\wh{\mu}(t) = \textbf{B}(t)^\top \wh{\bm{\mu}}$.
\item[Step 3:] Start with $k = 1$,
\begin{itemize}
\item[(1)] For $i = 1, \ldots, n$, update $\wh{\xi}_{ik}$ using (\ref{xiest}).
\item[(2)] Repeat the computation in (\ref{thetaest}) until the convergence of $\wh{\bm{\theta}}_k$.
\item[(3)] Repeat Step 3(1)--(2) until convergence.
\item[(4)] If $k < p$, let $k = k + 1$, repeat Step 3(1)--(3).
\end{itemize}
\item[Step 4:] Let $\wh{\psi}_k(t) = \textbf{B}(t)^\top \wh{\bm{\theta}}_k, k = 1, \ldots, p$, then $\wh{\phi}_k(t) = \wh{\psi}_k(t)/\|\wh{\psi}_k(t)\|_2$, where $\|\wh{\psi}_k(t)\|_2 = \{ \int_{\mathcal{T}} \wh{\psi}_k^2(t) dt \}^{1/2}$. Rescale $\wh{\bm{\xi}}_k$ correspondingly.
\end{itemize}

Let $\bm{\mu}^{(0)}, \bm{\xi}^{(0)}$ and $\bm{\Theta}^{(0)}$ denote the initial values of $\bm{\mu}, \bm{\xi}$ and $\bm{\Theta}$ respectively. Generally, one can set the initial values in a random way. Alternatively, we set $\bm{\mu}^{(0)}$ and $\bm{\Theta}^{(0)}$ as the FPCA estimates for $\{q_{ij}; i = 1,\ldots, n, j = 1, \ldots, m_i\}$ using local linear smoother \citep{yao2005functional}, neglecting the fact that these observations are binary, and then generate $\bm{\xi}^{(0)}$ randomly using the estimated eigenvalues. Throughout this article, we implement the latter scheme in initialization. Furthermore, choice for the number of FPCs is a long-standing issue in FPCA. Some popular information criterion, such as Akaike information crierion (AIC) \citep{yao2005functional} and BIC, can be applied. Note that as the FPCs are estimated sequentially in our algorithm, the number of FPCs has little effect on the SLFPCA estimates.

\subsection{Selection of Tuning Parameters}\label{SecTuning}

We take into account the selection of three tuning parameters involved in (\ref{penloglik1}): the smoothing parameter $\kappa_{\bm{\mu}}$ of mean function, the smoothing parameter $\kappa_{\bm{\theta}}$ of eigenfunctions, and the parameter $\lambda$ that controls the sparseness of eigenfunctions.

First, $\kappa_{\bm{\mu}}$ is selected via generalized cross-validation (GCV) method. In specific, $\kappa_{\bm{\mu}}$ only makes sense in the estimation of $\bm{\mu}$ in (\ref{muest}), which can be regard as smoothing $\{\wt{z}_{ij}; i = 1,\ldots, n, j = 1, \ldots, m_i\}$ through the penalized sum of squared errors fitting criterion. The details about GCV for the smoothing problem are provided in \citet{ramsay2005functional}. Next, we consider $\kappa_{\bm{\theta}}$ and $\lambda$ jointly as these two tuning parameters cooperate with each other in Step 3 of the algorithm. We define the following BIC-type criterion for the selection,
\begin{align}\label{bic}
\mbox{BIC}(\kappa_{\bm{\theta}}, \lambda) = -2 \sum_{i = 1}^n \sum_{j = 1}^{m_i} \log \pi \{q_{ij} (\textbf{B}_{ij}^\top \wh{\bm{\mu}} + \textbf{B}_{ij}^\top \wh{\bm{\Theta}}^\top \wh{\bm{\xi}}_i) \} + \Big (\sum_{k = 1}^p df_k \Big ) \cdot \log N,
\end{align}
where $df_k$ stands for the degrees of freedom in estimating $\bm{\theta}_k$. For a given $k$, let $\mathcal{A}_k$ be a set indexing non-zero elements in $\wh{\bm{\theta}}_k$.
Then
\begin{align}
df_k = \mbox{tr} \Big [U_{\mathcal{A}_k} \big \{ U_{\mathcal{A}_k}^{\top} U_{\mathcal{A}_k} + N \kappa_{\bm{\theta}} V_{\mathcal{A}_k} \big \}^{-1} U_{\mathcal{A}_k}^{\top} \Big ]. \nonumber
\end{align}
In practice, we select $(\kappa_{\bm{\theta}}, \lambda)$ that minimizes (\ref{bic}) from a set of candidates.

\section{ Theoretical Results }\label{SecTheory}

In this section, we study the consistency and sparsistency of the proposed method. We first discuss properties of $\wh{\phi}_k (t)$. Let $\mbox{NULL}(f) = \{t \in \mathcal{T} : f(t) = 0\}$ and $\mbox{SUPP}(f) = \{t \in \mathcal{T} : f(t) \neq 0\}$. The assumptions needed are listed as follows:
\begin{assumption}\label{Ass_phi}
There exists some constant $c > 0$ such that $|\phi_k^{(p^{\prime})} (t_1) - \phi_k^{(p^{\prime})} (t_2)| \leq c |t_1 - t_2|^{\nu}, \nu \in [0, 1]$. Moreover, $3/2 < r \leq d$, where $r = p^{\prime} + \nu$ and $d$ is the degree of the B-spline basis.
\end{assumption}
\begin{assumption}\label{Ass_sparsepen}
The tuning parameter $\lambda$ varies with $N$, and we assume that $\sqrt{\int_{SUPP_k} p_{\lambda}^{\prime} (|\phi_k (t)|)^2 dt} = O(N^{-1/2} K^{-3/2})$ and $\sqrt{\int_{SUPP_k} p_{\lambda}^{\prime \prime} (|\phi_k (t)|)^2 dt} = o(K^{-3/2})$ as $\lambda$ goes to zero, where $\mbox{SUPP}_k = \mbox{SUPP}(\phi_k)$ and $K$ is the number of interior knots for the B-spline basis.
\end{assumption}
\begin{assumption}\label{Ass_tuning}
For the number of interior knots, we assume $K = o(N^{1/4})$ and $K/N^{\frac{1}{2(r + 1)}} \rightarrow \infty$. For smoothing parameters, we assume $\kappa_{\bm{\mu}} = o(N^{-1/2})$ and $\kappa_{\bm{\theta}} = o(N^{-1/2})$. For sparseness parameter, we assume $\lambda = o(1)$ and $\lambda N^{1/2} K^{-3/2} \rightarrow \infty$.
\end{assumption}

Assumption \ref{Ass_phi} requires the eigenfunctions to be sufficiently smooth and refers to (H.3) in \citet{Cardot2003spline} and (C2) in \citet{lin2017locally}.
Assumption \ref{Ass_sparsepen} can be regarded as a functional generalization of ($\mbox{B}^{\prime}$) and ($\mbox{C}^{\prime}$) in \citet{fan2004nonconcave} and is the same as (C3) in \citet{lin2017locally}. This assumption ensures that the influence of the sparseness penalty on the estimation can be dominated by that of the likelihood function.
Assumption \ref{Ass_tuning} specifies the choosing condition for tuning parameters, which can be a guideline in the parameter selection.

\begin{theorem}\label{theoryconsis}
Under Assumptions \ref{Ass_phi} - \ref{Ass_tuning}, for $k = 1, \ldots, p$,
\begin{align}
\sup_{t \in \mathcal{T}} |\wh{\phi}_k (t) - \phi_k(t)| = O_p(N^{-1/2} K), \nonumber
\end{align}
when FPC score $\bm{\xi}_0$ is given.
\end{theorem}

\begin{theorem}\label{theorySpar}
Under Assumptions \ref{Ass_phi} - \ref{Ass_tuning}, $\mbox{NULL}(\wh{\phi}_k) \rightarrow \mbox{NULL}(\phi_k)$ and $\mbox{SUPP}(\wh{\phi}_k) \rightarrow \mbox{SUPP}(\phi_k)$ in probability, as $N \rightarrow \infty$, when FPC score $\bm{\xi}_0$ is given.
\end{theorem}

Consistency and sparsistency of the estimated eigenfunctions are stated in the above two theorems. We then explore the asymptotic property of the estimated FPC scores. The following assumption is required:

\begin{assumption}\label{Ass_M}
\normalfont The observation sizes $m_i$'s are independent realizations of the random variable $m$, and are independent of $\big \{(t_{ij}, Y_{ij}): j = 1, \ldots, m_i \big \}$.
Assume that $m_i = O_p(M), i = 1, \ldots, N$ and $M \rightarrow \infty$.
\end{assumption}

\begin{theorem}\label{theoryXi}
Under Assumption \ref{Ass_M}, we have
\begin{align}
|\wh{\xi}_{ik} - \xi_{ik}| = O_p(M^{-1/2}), \nonumber
\end{align}
for $i = 1, \ldots, n, k = 1, \ldots, p$, when coefficient matrix $\bm{\Theta}_0$ of eigenfunctions is given.
\end{theorem}

\begin{remark}
\normalfont In fact, the simultaneous derivation of asymptotic properties for the estimated eigenfunctions $\wh{\phi}_k(t)$ and FPC scores $\wh{\xi}_{ik}$ is intractable, because of the large number of parameters. Hence, we discuss their properties separately. However, the above theoretical results can still bring some insights for the estimates. It shows that the imposed penalties would not lead to invalid results, and the sparseness penalty is effective in identifying non-zero subinterval for our problem. Moreover, the simulation studies in Section \ref{SecSim} further demonstrate the good performance of the SLFPCA method in practice.
\end{remark}

\begin{remark}
\normalfont Theorem \ref{theoryXi} implies that the convergence rate of the estimated FPC scores depends on the observation size $M$. In specific, a divergent observation size is needed for the consistency of FPC scores, while the requirement is dispensable for asymptotic properties of the estimated eigenfunctions, as shown in Theorem \ref{theoryconsis} and Theorem \ref{theorySpar}. It is quite natural as FPC scores are varied from individual to individual, while eigenfunctions are specific for all subjects.
\end{remark}

\section{ Simulation }\label{SecSim}

In this section, we conduct a comprehensive numerical study to evaluate the performance of our SLFPCA method. We consider two scenarios in our simulation. First, we set functions that being strictly zero-valued in some subintervals as the true eigenfunctions. Second, the true eigenfunctions are set to be non-zero almost in the whole interval. We compare our SLFPCA method with the LGP method in \citet{hall2008modelling}. The criteria for assessment are as follows:
\begin{align}
\mbox{ISE}_{\mu} &= \| \mu - \wh{\mu} \|^2 = \int_{\mathcal{T}} \{\mu(t) - \wh{\mu}(t)\}^2 dt, \nonumber \\
\mbox{ISE}_k &= \| \phi_k - \wh{\phi}_k \|^2 = \int_{\mathcal{T}} \{\phi_k(t) - \wh{\phi}_k(t)\}^2 dt, \nonumber
\end{align}
where $\mbox{ISE}_{\mu}$ and $\mbox{ISE}_k$ measure the error of mean function estimates and eigenfunction estimates respectively. Obviously, a lower $\mbox{ISE}_{\mu}$ or $\mbox{ISE}_k$ indicates a more precise estimate.

\subsection{Sparse FPCs}\label{SubSecSparFPC}

We first discuss the behaviours of SLFPCA and LGP methods when the true FPCs possess local sparse features. To generate binary data $\{y_{ij}\mbox{:} \ i = 1,\ldots, n, j = 1, \ldots, m_i\}$, we begin with constructing $n = 200$ independent latent processes $X_i(t)$ through (\ref{KL}). The latent processes have mean function $\mu(t) = 2 \cdot \mbox{sin}(\pi t/5)/\sqrt{5}, t \in [0, 10]$. For the eigenfunctions $\phi_k(t)$, let $B_l(t)$ denote the $l$-th B-spline basis on [0, 10], with degree three and nine equally spaced interior knots. We explore the following two cases:
\begin{itemize}
\item Case 1: Define $\psi_1(t) = B_4(t), \psi_2(t) = B_{10}(t)$, then $\phi_k(t) = \psi_k(t)/\|\psi_k\|_2$, $k = 1, 2$.
\item Case 2: Define $\psi_1(t) = B_7(t), \psi_2(t) = B_4(t) - B_{10}(t)$, then $\phi_k(t) = \psi_k(t)/\|\psi_k\|_2$, $k = 1, 2$.
\end{itemize}
Moreover, the eigenvalues are set as $\lambda_1 = 3^2$, $\lambda_2 = 2^2$ and $\lambda_k = 0, k \geq 3$. The FPC scores $\xi_{ik}$ are simulated from $\mathcal{N}(0, \lambda_{k})$. Finally, we yield $y_{ij}$ using the probability obtained from $X_i(t)$. With regard to the observation grids, as our method can be executed to both dense and sparse (or longitudinal) designs, we take into account these two various designs in our simulation. Specifically, for the dense design, we consider the regular case and the observation sizes for all subjects are set as $m_1 = \cdots = m_n = 51$. On the other hand, for the sparse design, $m_i$ is uniformly selected from $\{8, \ldots, 12\}$ and observation grids $t_{ij}$ are uniformly sampled from $[0, 10]$ corresponding to $m_i$. We report the results for the dense design here and relegate the analysis for sparse design in the Supplementary Material. In addition, we also consider the settings in which only the first eigenvalue is non-zero as \citet{hall2008modelling} and present the results in the Supplementary Material.

Table \ref{SimTabSparse} lists the simulation results of SLFPCA and LGP over 100 Monte Carlo runs for the dense design when the true FPCs with local sparse features are provided. For the two considered cases that accept various FPC settings, it is evident that SLFPCA achieves much smaller $\mbox{ISE}_1$ and $\mbox{ISE}_2$, which implies SLFPCA outperforms LGP on the estimation of eigenfunctions. The attractive performance of SLFPCA compared with LGP is in accordance with our expectation, as the sparse FPCs scenario here is in favor of our method. Moreover, these two methods are at a similar level in estimating the mean functions, for they get nearly the same $\mbox{ISE}_{\mu}$. Schematically, Figures \ref{SimFigSparse1} and \ref{SimFigSparse2} exhibit the estimated FPCs in one randomly chosen run for both Case 1 and Case 2 respectively. It is shown that the estimated eigenfunctions obtained from LGP are non-zero over almost the whole interval, while SLFPCA can correctly identify the subinterval on which the true FPCs are non-zero valued.
As SLFPCA owns a great capacity in capturing the local sparse features, it is natural that SLFPCA gains more promising $\mbox{ISE}_1$ and $\mbox{ISE}_2$.

\begin{table}[H]
\caption{Average $\mbox{ISE}_{\mu}$ and $\mbox{ISE}_k$ with standard deviation in parentheses for 100 Monte Carlo runs, when the true FPCs are sparse.}
\label{SimTabSparse}
\begin{center}
\begin{tabular}{ccccc}
\hline
 &Method&$\mbox{ISE}_{\mu}$&$\mbox{ISE}_1$&$\mbox{ISE}_2$\\
\hline
\multirow{2}{*}{Case 1}&SLFPCA&0.3632(0.1472)&0.0182(0.0143)&0.0172(0.0131)\\
 &LGP&0.3653(0.1230)&0.1142(0.0624)&0.1064(0.0617)\\
\hline
\multirow{2}{*}{Case 2}&SLFPCA&0.1541(0.0805)&0.0455(0.1308)&0.0475(0.1211)\\
 &LGP&0.1627(0.0834)&0.3551(0.4712)&0.3319(0.4804)\\
\hline
\end{tabular}
\end{center}
\end{table}

\begin{figure}[H]
  \centering
  \includegraphics[width=\textwidth]{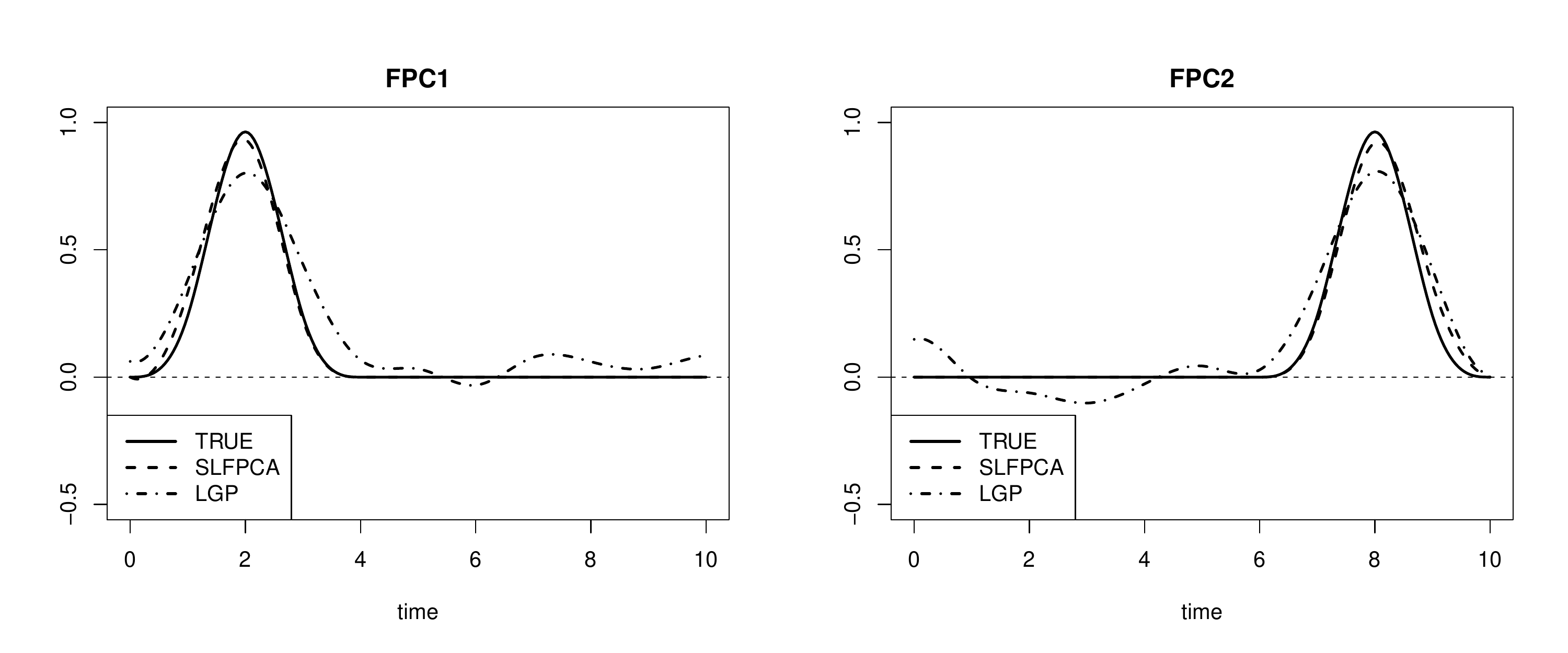}\\
  \caption{Eigenfunction estimates in one randomly chosen run for Case 1 illustrated in Section \ref{SubSecSparFPC} when the true FPCs are sparse. Thick solid lines are for true eigenfunctions, while lines for the estimated eigenfunctions obtained from SLFPCA and LGP are in dashed and dotted-dashed types respectively.}
  \label{SimFigSparse1}
\end{figure}

\begin{figure}[H]
  \centering
  \includegraphics[width=\textwidth]{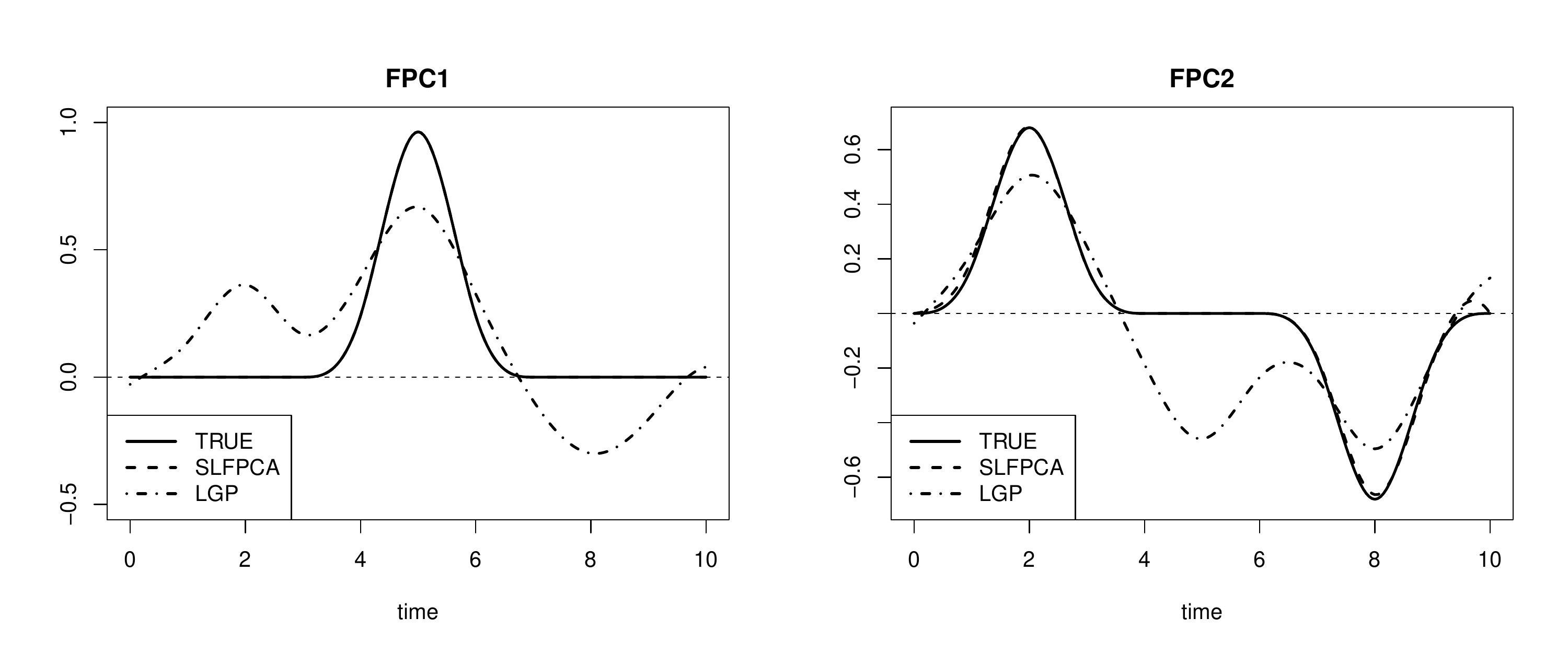}\\
  \caption{Eigenfunction estimates in one randomly chosen run for Case 2 illustrated in Section \ref{SubSecSparFPC} when the true FPCs are sparse. Thick solid lines are for true eigenfunctions, while lines for the estimated eigenfunctions obtained from SLFPCA and LGP are in dashed and dotted-dashed types respectively.}
  \label{SimFigSparse2}
\end{figure}

\subsection{Non-sparse FPCs}\label{SubSecNonSparFPC}

We then explore the simulation results when the true FPCs are non-sparse. This scenario is not inclined to SLFPCA any more, whereas we shall show our method still yields nice estimating results compared with LGP. The setups are the same as that in Section \ref{SubSecSparFPC}, except for eigenfunctions. We also consider two cases:
\begin{itemize}
\item Case 3: $\phi_1(t) = \mbox{cos}(\pi t/5)/\sqrt{5}, \phi_2(t) = \mbox{sin}(\pi t/5)/\sqrt{5}, t \in [0,10]$.
\item Case 4: $\phi_1(t) = \mbox{cos}(\pi t/5)/\sqrt{5}, \phi_2(t) = \mbox{cos}(2 \pi t/5)/\sqrt{5}, t \in [0,10]$.
\end{itemize}
These two cases have the same first eigenfunction, while the second eigenfunction in Case 4 undertakes more variability. We present the estimating results for dense design here. The remaining results, for sparse design and for the case where only the first eigenvalue is non-zero, are provided in the Supplementary Material.

The simulation results over 100 Monte Carlo runs when the true FPCs are non-sparse are displayed in Table \ref{SimTabNonSparse}. It is observed that SLFPCA still reaches lower $\mbox{ISE}_1$ and $\mbox{ISE}_2$ than LGP, though the difference between these two methods is much smaller than that in Section \ref{SubSecSparFPC}. SLFPCA also offers a more accurate estimate for the mean function according to $\mbox{ISE}_{\mu}$. Therefore, SLFPCA is a competitive approach even when the true FPCs show no local sparse feature. Further, the estimated eigenfunctions are visualized in Figures \ref{SimFigNonSparse1} and \ref{SimFigNonSparse2} for one randomly chosen run. Both figures clarify that SLFPCA and LGP perform similarly when true FPCs are non-sparse and yield estimates close to the true eigenfunctions. Note that SLFPCA does not produce sparse eigenfunction estimates in Figures \ref{SimFigNonSparse1} and \ref{SimFigNonSparse2}. The reason is that the tuning parameter $\lambda$ is selected to be zero via BIC, and SLFPCA is equivalent to general FPCA when $\lambda = 0$. Hence, these two cases further demonstrate the ability of SLFPCA in identifying the non-zero subintervals.

\begin{table}[H]
\caption{Average $\mbox{ISE}_{\mu}$ and $\mbox{ISE}_k$ with standard deviation in parentheses for 100 Monte Carlo runs, when the true FPCs are non-sparse.}
\label{SimTabNonSparse}
\begin{center}
\begin{tabular}{ccccc}
\hline
 &Method&$\mbox{ISE}_{\mu}$&$\mbox{ISE}_1$&$\mbox{ISE}_2$\\
\hline
\multirow{2}{*}{Case 3}&SLFPCA&0.2441(0.1292)&0.0151(0.0171)&0.0175(0.0172)\\
 &LGP&0.2678(0.1302)&0.0178(0.0181)&0.0218(0.0180)\\
\hline
\multirow{2}{*}{Case 4}&SLFPCA&0.1955(0.0715)&0.0113(0.0125)&0.0270(0.0170)\\
 &LGP&0.2075(0.0698)&0.0168(0.0145)&0.0298(0.0178)\\
\hline
\end{tabular}
\end{center}
\end{table}

\begin{figure}[H]
  \centering
  \includegraphics[width=\textwidth]{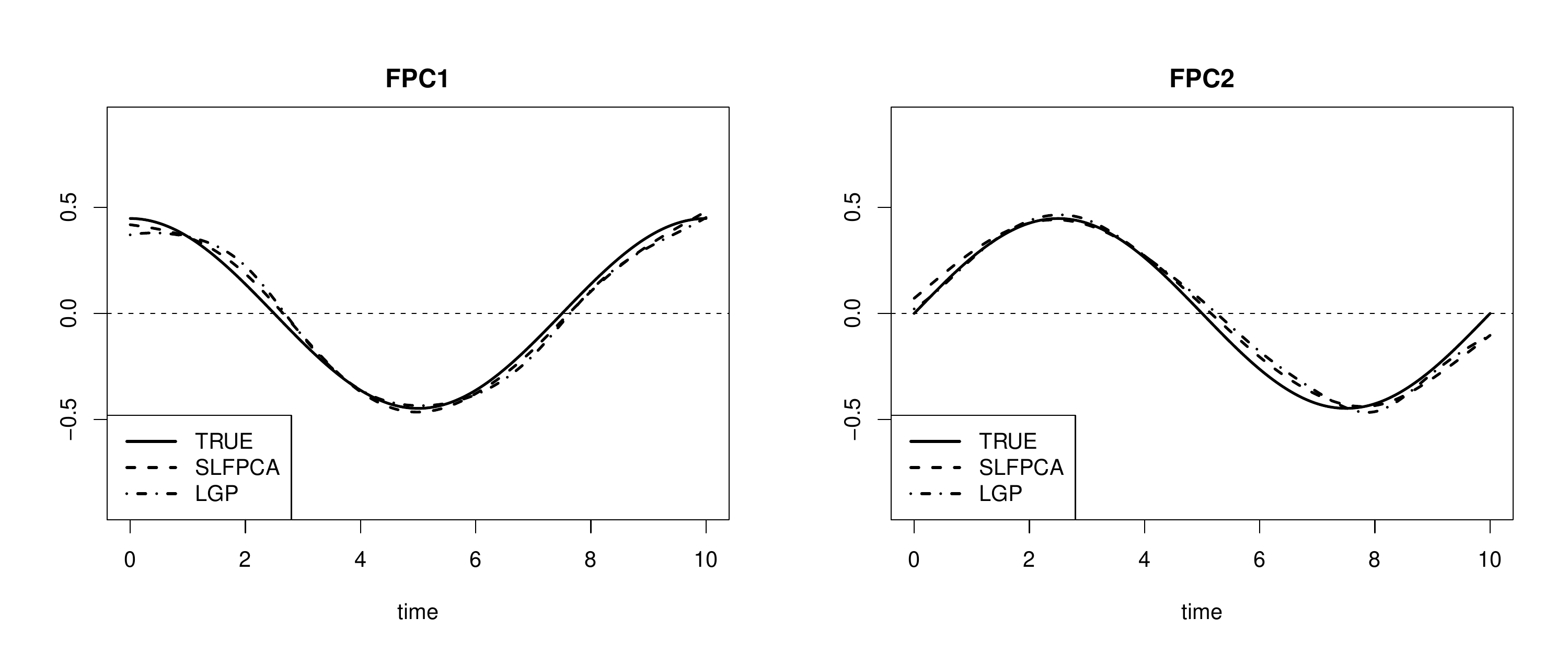}\\
  \caption{Eigenfunction estimates in one randomly chosen run for Case 3 illustrated in Section \ref{SubSecNonSparFPC} when the true FPCs are non-sparse. Thick solid lines are for true eigenfunctions, while lines for the estimated eigenfunctions obtained from SLFPCA and LGP are in dashed and dotted-dashed types respectively.}
  \label{SimFigNonSparse1}
\end{figure}

\begin{figure}[H]
  \centering
  \includegraphics[width=\textwidth]{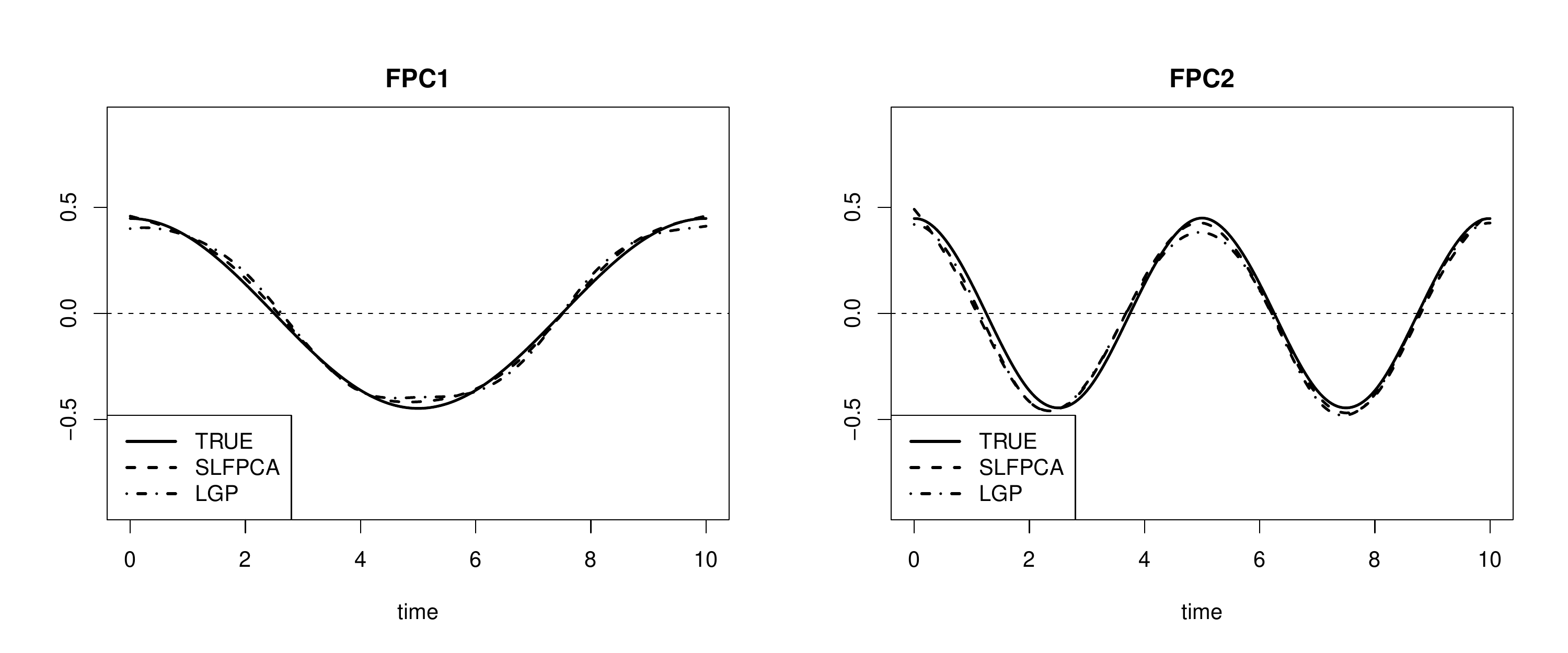}\\
  \caption{Eigenfunction estimates in one randomly chosen run for Case 4 illustrated in Section \ref{SubSecNonSparFPC} when the true FPCs are non-sparse. Thick solid lines are for true eigenfunctions, while lines for the estimated eigenfunctions obtained from SLFPCA and LGP are in dashed and dotted-dashed types respectively.}
  \label{SimFigNonSparse2}
\end{figure}

\section{ Real Data Analysis }\label{SecReal}

In this section, we apply our proposed SLFPCA method to the physical activity data collected from \citet{kozey-keadle2014changes}. The data are generated from a health monitoring project which measured the metabolic effects of several interventions to increase physical activity and reduce sedentary behaviors (e.g. sitting or lying down) in office workers. A wearable monitor, ActivPAL$^{\rm TM}$ (www.paltech.plus.com), was used to track the wearer's leg movement over time. In particular, the device detected leg angle change when the wearer stands up, which showed an interruption of sedentary behavior (0, no; 1, yes). For each participant, the observations obtained from the monitor are summarized into consecutive five-minute intervals. There are $n = 60$ individuals involved in this project and each individual was tracked for $36$ five-minute records.

Figures \ref{realmean} and \ref{realeig} show the estimated mean function $\wh{\mu}(t)$ and eigenfunctions $\wh{\phi}_k(t)$'s by SLFPCA. The tuning parameters are selected as presented in Section \ref{SecTuning}. Moreover, we choose the number of FPCs as $p = 2$ by BIC. The mean function for the latent process indicates individuals were likely to interrupt their sedentary behaviors to take intense exercises at about $t = 12$. After about 30 minutes' active physical exercises, more sedentary behaviors were observed and then the interruptions of sedentary behaviors increased back to the starting level.

\begin{figure}[H]
  \centering
  \includegraphics[width=0.7 \textwidth]{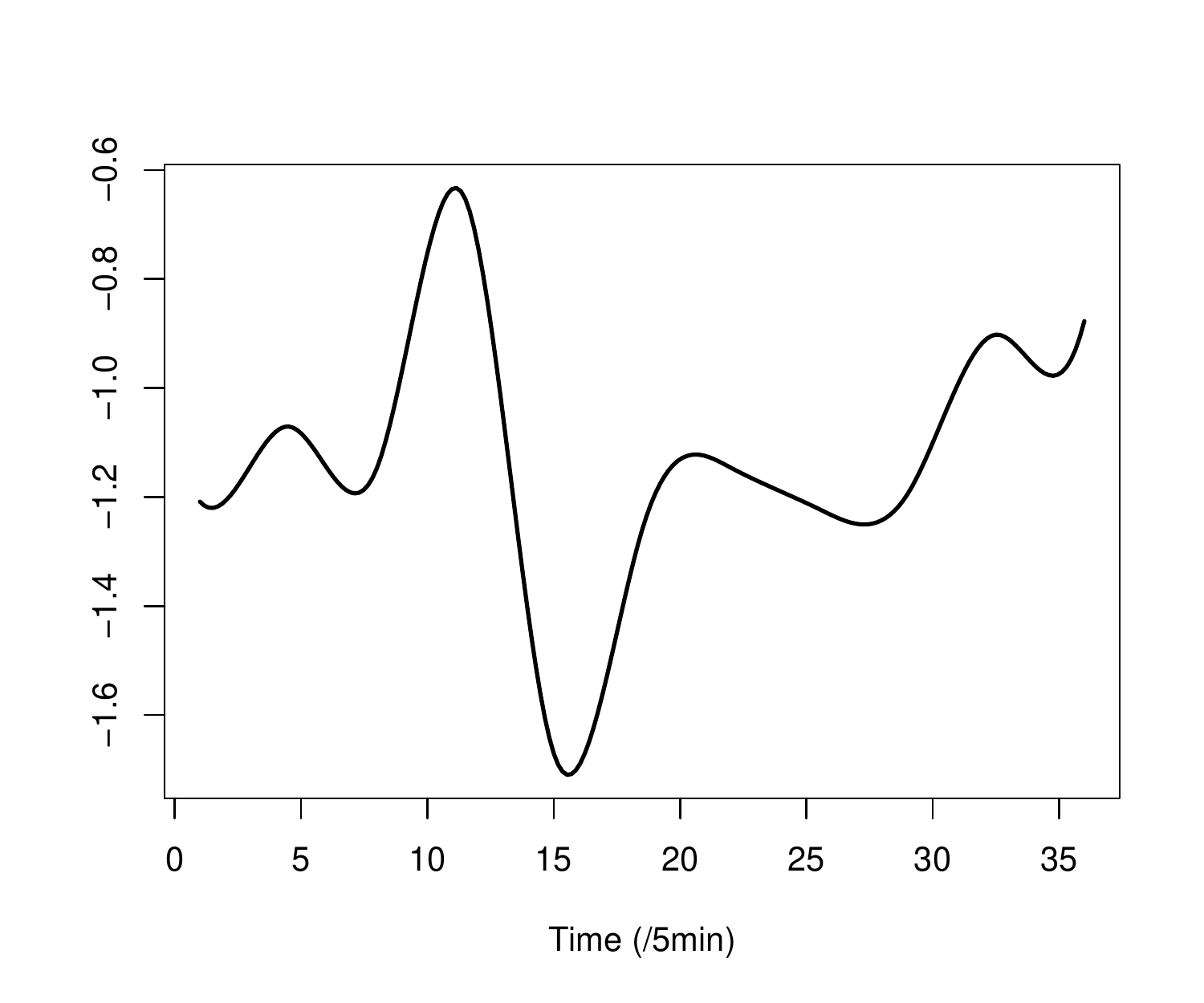}\\
  \caption{ Estimated mean function of the latent process $X(t)$ for the physical activity data.}
  \label{realmean}
\end{figure}

\begin{figure}[htbp]
  \centering
  \includegraphics[width=\textwidth]{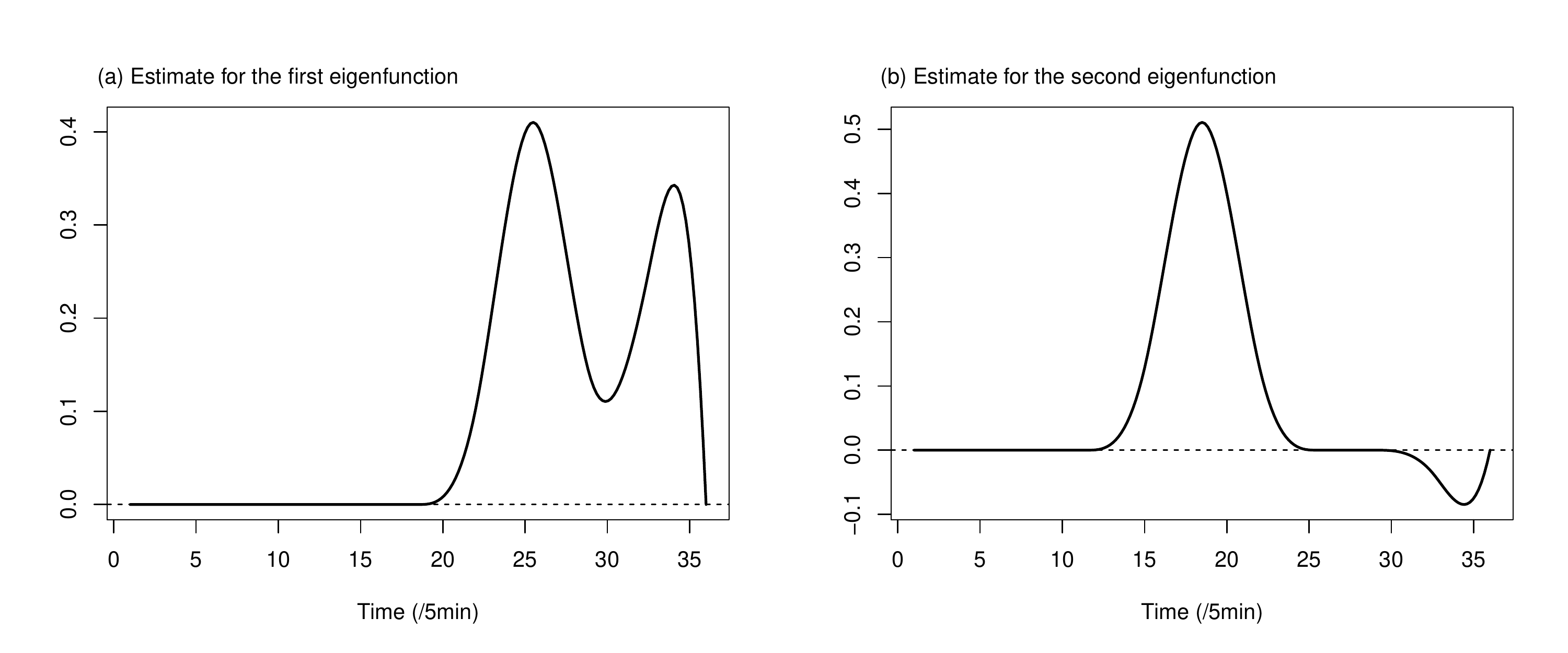}\\
  \caption{The estimates of the first two eigenfunctions for the physical activity data.}
  \label{realeig}
\end{figure}

The estimated eigenfunctions reveal some local sparse features, which facilitate interpretation for the results. The first eigenfunction highlights the variation after $t = 20$, while being zero on the remaining part. It implies that individuals experienced similar activity pattern with each other on $[0, 20]$, as the variation around mean function at that interval is ignorable. On the other hand, the variation after $t = 20$ is substantial, especially at $t = 25$ and $t = 34$. In this physical activity study, most of the participants started to take a one-hour exercises at $t=12$ based on training schedule. Thus, the variation after $t=20$ showed different activity pattern after about 30 minutes' intense exercises. Some individuals were still active with frequent interruption of sedentary behaviors, while others preferred sitting for a long time to have a rest. For the second eigenfunction, it is positive on $[15, 25]$ and it turns to be negative on $[30, 35]$, which indicates a negative association relationship between the observations on these two intervals. A possible explanation is that, the subject who were more active during the exercises with high frequency of sedentary behavior interruptions may have longer sitting time after exercises.

\section{ Conclusion and Discussion }\label{SecCon}

In this paper, we introduce a novel SLFPCA method for functional binary data and require the estimated FPCs to be able to capture the local sparse features of the original FPCs for the sake of interpretability. To this end, we construct a penalized Bernoulli likelihood with both roughness penalty and sparseness penalty. The sparseness penalty is crucial for the realization of local sparsity and we generalize the fSCAD penalty to our FPCA issue. The simulation study shows the superiority of SLFPCA and illustrates its encouraging identifying ability for non-zero subintervals. The practical application to the physical activity data suggests SLFPCA actually helps the interpretation a lot.

As it is the first try on sparse FPCA for binary data, there exists plenty of extensions in relevant field. First, other sparseness penalties, such as group bridge penalty \citep{wang2015functional, tu2020estimation} and LASSO penalty \citep{centofanti2020smooth}, can also be extended to the sparse FPCA problem. And it may be an interesting affair to explore the influences of choosing various sparseness penalties. Second, we presume identical sparseness tuning parameters for all considered eigenfunctions in our work. There may be cases where eigenfunctions meet different sparsity and thus distinct sparseness tuning parameters are necessary. The number of tuning parameter increases for these cases and such multiple tuning parameter selections would consume much computation time. Hence, a more effective method for selecting tuning parameters is in need. Third, the idea in this paper can be applied to functional data from other discrete distributions, such as Poisson distribution for functional count data, through altering the penalized likelihood corresponding to the distribution. It is worthwhile to develop adaptive algorithms for diverse distributions.

\section*{Acknowledgements}

This research was supported by Public Health $\&$ Disease Control and Prevention, Major Innovation $\&$ Planning Interdisciplinary Platform for the ``Double-First Class" Initiative, Renmin University of China.

\appendix
\section{ Proofs }

For simplicity of notation, we neglect mean function $\mu(t)$ here, that is $\bm{\mu} = 0$. The proofs can be easily generalized to the cases where $\bm{\mu} \neq \bm{0}$.

\subsection{Proof of Theorem \ref{theoryconsis}}

\begin{lemma}\label{lemma_Bspline}
(Approximation properties of B-splines) Assume function $f(t)$ satisfying $|f^{(p^{\prime})}(t_1) - f^{(p^{\prime})}(t_2)| \leq c |t_1 - t_2|^{\nu}, c > 0, \nu \in [0, 1]$. Then there is some $\wt{f}(t) = \sum_{l = 1}^L b_l B_l(t)$ such that $\|\wt{f}(t) - f(t)\|_{\infty} = O(K^{-r})$, where $\|\wt{f}(t) - f(t)\|_{\infty} = \sup_{t \in \mathcal{T}} |\wt{f}(t) - f(t)|$ and $r = p^{\prime} + \nu$.
\end{lemma}

\begin{proof}[Proof of Lemma \ref{lemma_Bspline}]
Let $w(f; h) = \sup \{|f(t) - f(s)|: t, s \in \mathcal{T}, |t - s| \leq h \}$. According to Theorem \RNum{12} (6) in \citet{boor1978a}, there exists some $\wt{f}(t) = \sum_{l = 1}^L b_l B_l(t)$ such that
\begin{align}
\|\wt{f} - f\|_{\infty} \leq C_0 \cdot h^{p^{\prime}} \cdot w(f^{(p^{\prime})}; h),  \nonumber
\end{align}
where $h$ is the distance between the adjacent knots, thus $h = O(K^{-1})$. Further, as $|f^{(p^{\prime})}(t_1) - f^{(p^{\prime})}(t_2)| \leq c |t_1 - t_2|^{\nu}, c > 0, \nu \in [0, 1]$, we have $f^{(p^{\prime})}$ satisfies a H\"{o}lder condition with component $\nu$. Hence, according to Theorem \RNum{2} (21) in \citet{boor1978a}, we have
\begin{align}
w(f^{(p^{\prime})}; h) \leq C_1 h^{\nu}, \nonumber
\end{align}
where $C_1$ is some constant. Therefore,
\begin{align}
\|\wt{f} - f\|_{\infty} \leq C_0 C_1 h^{p^{\prime} + \nu} = C_0 C_1 h^{r} = O(K^{-r}).  \nonumber
\end{align}
The proof is completed.
\end{proof}

\begin{proof}[Proof of Theorem \ref{theoryconsis}]

Let $\bm{\Omega} = (\bm{\theta}_1^\top, \ldots, \bm{\theta}_p^\top)^\top$. Then our objective function is equivalent to
\begin{align}
\wt{Q}(\bm{\Omega}, \bm{\xi}_0) = -\frac{1}{N} L(\bm{\Omega}, \bm{\xi}_0) + \kappa_{\bm{\theta}} \sum_{k = 1}^p \bm{\theta}_k^\top V \bm{\theta}_k + \mbox{PEN}_{\lambda}(\bm{\Theta}), \nonumber
\end{align}
where
\begin{align}
L(\bm{\Omega}, \bm{\xi}_0) = \sum_{i = 1}^n \sum_{j = 1}^{m_i} \log \pi \{q_{ij} (\textbf{B}_{ij}^\top \bm{\Theta}^\top \bm{\xi}_{0i}) \}. \nonumber
\end{align}
Let $\alpha_N = N^{-1/2} K$.We want to show that for any $\epsilon > 0$, $\exists C_1 > 0$, $\forall C > C_1$,
\begin{align}\label{theorygoal}
P \Big \{ \inf_{\|\bm{u}\|_2 = C} \wt{Q}(\bm{\Omega}_0 + \alpha_N \bm{u}, \bm{\xi}_0) > \wt{Q}(\bm{\Omega}_0, \bm{\xi}_0) \Big \} \geq 1 - \epsilon,
\end{align}
where $\bm{\Omega}_0 = (\bm{\theta}_{01}^\top, \ldots, \bm{\theta}_{0p}^\top)^\top$ is the true parameter.
It indicates there exists a local minimizer in the ball $\{\bm{\Omega}_0 + \alpha_N \bm{u}: \|\bm{u}\|_2 \leq C\}$, with probability at least $1 - \epsilon$. Moreover, the local minimizer satisfies $\|\wh{\bm{\Omega}} - \bm{\Omega}_0\|_2 = O_p (\alpha_N)$, where $\wh{\bm{\Omega}} = (\wh{\bm{\theta}}_1^\top, \ldots, \wh{\bm{\theta}}_p^\top)^\top$.

In specific, let $\bm{u} = (\bm{u}_1^\top, \ldots, \bm{u}_p^\top)^\top$ and $\wt{\phi}_k(t) = \bm{B}^\top (t) \bm{\theta}_{0k}$. We have
\begin{align}
D_N (\bm{u}) \triangleq& \wt{Q}(\bm{\Omega}_0 + \alpha_N \bm{u}, \bm{\xi}_0) - \wt{Q}(\bm{\Omega}_0, \bm{\xi}_0) \nonumber \\
=& -\frac{1}{N} \{L(\bm{\Omega}_0 + \alpha_N \bm{u}, \bm{\xi}_0) - L(\bm{\Omega}_0, \bm{\xi}_0)\} \nonumber \\
&+ \kappa_{\bm{\theta}} \sum_{k = 1}^p \{(\bm{\theta}_{0k} + \alpha_N\bm{u}_k)^\top V (\bm{\theta}_{0k} + \alpha_N\bm{u}_k) - \bm{\theta}_{0k}^\top V \bm{\theta}_{0k} \} \nonumber \\
&+ \{\mbox{PEN}_{\lambda} (\bm{\Theta}_{0 \alpha_N}) - \mbox{PEN}_{\lambda} (\bm{\Theta}_0)\} \nonumber \\
=& -\frac{1}{N} \{L(\bm{\Omega}_0 + \alpha_N \bm{u}, \bm{\xi}_0) - L(\bm{\Omega}_0, \bm{\xi}_0)\} \nonumber \\
&+ \kappa_{\bm{\theta}} \sum_{k = 1}^p \{(\bm{\theta}_{0k} + \alpha_N\bm{u}_k)^\top V (\bm{\theta}_{0k} + \alpha_N\bm{u}_k) - \bm{\theta}_{0k}^\top V \bm{\theta}_{0k} \} \nonumber \\
&+ \frac{K + 1}{8T} \sum_{k = 1}^{p} \int_{\mathcal{T}} \Big \{ p_{\lambda} \Big ( \big |\bm{B}^\top (t) (\bm{\theta}_{0k} + \alpha_N \bm{u}_k) \big | \Big ) - p_{\lambda} \Big ( \big |\bm{B}^\top (t) \bm{\theta}_{0k} \big | \Big ) \Big \} dt \nonumber \\
\geq& -\frac{1}{N} \{L(\bm{\Omega}_0 + \alpha_N \bm{u}, \bm{\xi}_0) - L(\bm{\Omega}_0, \bm{\xi}_0)\} \nonumber \\
&+ \kappa_{\bm{\theta}} \sum_{k = 1}^p \{(\bm{\theta}_{0k} + \alpha_N\bm{u}_k)^\top V (\bm{\theta}_{0k} + \alpha_N\bm{u}_k) - \bm{\theta}_{0k}^\top V \bm{\theta}_{0k} \} \nonumber \\
&+ \frac{K + 1}{8T} \sum_{k = 1}^{p} \int_{\mathcal{G}_k} \Big \{ p_{\lambda} \Big ( \big |\bm{B}^\top (t) (\bm{\theta}_{0k} + \alpha_N \bm{u}_k) \big | \Big ) - p_{\lambda} \Big ( \big |\bm{B}^\top (t) \bm{\theta}_{0k} \big | \Big ) \Big \} dt \nonumber \\
\triangleq& \Delta_1 + \Delta_2 + \Delta_3, \nonumber
\end{align}
where $\bm{\Theta}_{0\alpha_N}$ is obtained from $\bm{\theta}_{0k} + \alpha_N \bm{u}_k, k = 1, \ldots, p$, and $\mathcal{G}_k = \mbox{SUPP}(\wt{\phi}_k)$.
For $\Delta_1$, according to the Taylor expansion, we have
\begin{align}
\Delta_1 = -\frac{1}{N} [\alpha_N \nabla^\top L(\bm{\Omega}_0, \bm{\xi}_0) \bm{u} + \frac{1}{2} \bm{u}^\top \nabla^2 L(\bm{\Omega}_0, \bm{\xi}_0) \bm{u} \alpha_N^2 \{1 + o_p(1)\}], \nonumber
\end{align}
where
\begin{align}
\Big |-\frac{1}{N} \alpha_N \nabla^\top L(\bm{\Omega}_0, \bm{\xi}_0) \bm{u} \Big | &\leq \frac{1}{N} \alpha_N \|\nabla^\top L(\bm{\Omega}_0, \bm{\xi}_0)\|_2 \|\bm{u}\|_2 \nonumber \\
&= O_p(\alpha_N N^{-1/2} K^{1/2}) \|\bm{u}\|_2, \nonumber
\end{align}
\begin{align}\label{dominate}
-\frac{1}{2N} \bm{u}^\top \nabla^2 L(\bm{\Omega}_0, \bm{\xi}_0) \bm{u} \alpha_N^2 = \frac{\alpha_N^2}{2} \bm{u}^\top I(\bm{\Omega}_0, \bm{\xi}_0) \bm{u} + o_p(\alpha_N^2).
\end{align}
The derivation of (\ref{dominate}) is provided in the proof of Lemma \ref{dominateDer}.
For $\Delta_2$, we have
\begin{align}
&\kappa_{\bm{\theta}}\{(\bm{\theta}_{0k} + \alpha_N \bm{u}_k)^\top V (\bm{\theta}_{0k} + \alpha_N \bm{u}_k) - \bm{\theta}_{0k}^\top V \bm{\theta}_{0k}\} \nonumber \\
=& 2 \kappa_{\bm{\theta}} \alpha_N \bm{\theta}_{0k}^\top V \bm{u}_k + \kappa_{\bm{\theta}} \alpha_N^2 \bm{u}_k^\top V \bm{u}_k \nonumber \\
\leq& 2 \kappa_{\bm{\theta}} \alpha_N \|\bm{\theta}_{0k}^\top\|_2 \|V\|_2 \|\bm{u}_k\|_2 + \kappa_{\bm{\theta}} \alpha_N^2 \|\bm{u}_k^\top \|_2 \| V \|_2 \| \bm{u}_k \|_2 \nonumber \\
=& o(N^{-1/2} \alpha_N K) \|\bm{u}_k\|_2 + o(N^{-1/2} \alpha_N^2 K) \|\bm{u}_k\|_2^2 \nonumber \\
=& o(N^{-1/2} K \alpha_N) \|\bm{u}_k\|_2. \nonumber
\end{align}
Therefore,
\begin{align}
\Delta_2 = o(N^{-1/2} K \alpha_N) \Big ( \sum_{k = 1}^p \|\bm{u}_k\|_2 \Big ). \nonumber
\end{align}
For $\Delta_3$, we have
\begin{align}
& \frac{K + 1}{8T} \sum_{k = 1}^{p} \int_{\mathcal{G}_k} \Big \{ p_{\lambda} \Big ( \big |\bm{B}^\top (t) (\bm{\theta}_{0k} + \alpha_N \bm{u}_k) \big | \Big ) - p_{\lambda} \Big ( \big |\bm{B}^\top (t) \bm{\theta}_{0k} \big | \Big ) \Big \} dt \nonumber \\
=& \frac{K + 1}{8T} \sum_{k = 1}^{p} \int_{\mathcal{G}_k} \Big [ \alpha_N \nabla^\top \Big \{p_{\lambda} \Big (\big |\bm{B}^\top (t) \bm{\theta}_{0k} \big | \Big ) \Big \} \bm{u}_k \nonumber \\
& \qquad \qquad \qquad \qquad \qquad + \frac{\alpha_N^2}{2} \bm{u}_k^\top \nabla^2 \Big \{p_{\lambda} \Big (\big |\bm{B}^\top (t) \bm{\theta}_{0k} \big | \Big ) \Big \} \bm{u}_k \Big ] dt \nonumber \\
=& \frac{K + 1}{8T} \alpha_N \sum_{k = 1}^p \nabla^\top \Big \{ \int_{\mathcal{G}_k} p_{\lambda} \Big (\big |\bm{B}^\top (t) \bm{\theta}_{0k} \big | \Big ) dt \Big \} \bm{u}_k \nonumber \\
& \qquad \qquad \qquad \qquad \qquad + \frac{K + 1}{8T} \frac{\alpha_N^2}{2} \sum_{k = 1}^p \bm{u}_k^\top \nabla^2 \Big \{ \int_{\mathcal{G}_k} p_{\lambda} \Big (\big |\bm{B}^\top (t) \bm{\theta}_{0k} \big | \Big ) dt \Big \} \bm{u}_k \nonumber \\
\leq& O_p(\alpha_N N^{-1/2} K^{-1/2}) \Big ( \sum_{k = 1}^p \|\bm{u}_k\|_2 \Big ) + o_p(\alpha_N^2 K^{-1}) \Big ( \sum_{k = 1}^p \bm{u}_k^\top Z \bm{u}_k \Big ), \nonumber
\end{align}
where $Z$ is a sparse matrix with $1$ in the location $(i, j)$ such that $0 \leq |i - j| \leq 4$. The derivation of the last inequality makes use of the results in \citet{lin2017locally}, that is
\begin{align}
\Big \| \nabla^\top \Big \{ \int_{\mathcal{G}_k} p_{\lambda} \Big (\big |\bm{B}^\top (t) \bm{\theta}_{0k} \big | \Big ) dt \Big \} \Big \|_2 = O(N^{-1/2} K^{-3/2}) \nonumber
\end{align}
\begin{align}
\nabla^2 \Big \{ \int_{\mathcal{G}_k} p_{\lambda} \Big (\big |\bm{B}^\top (t) \bm{\theta}_{0k} \big | \Big ) dt \Big \} = o(K^{-2}) Z. \nonumber
\end{align}
Allowing $\|\bm{u}\|_2$ to be large enough, all terms are dominated by the second term of $\Delta_1$. Therefore, we obtain (\ref{theorygoal}) according to (\ref{dominate}). Thus $\|\wh{\bm{\theta}}_k - \bm{\theta}_{0k}\|_2 = O_p(N^{-1/2} K), k = 1, \ldots, p$.

According to Lemma \ref{lemma_Bspline}, we have $\|\wt{\phi}_k - \phi_k\|_{\infty} = O(K^{-r})$. Then
\begin{align}
\|\wh{\phi}_k - \phi_k\|_{\infty} &\leq \|\wh{\phi}_k - \wt{\phi}_k\|_{\infty} + \|\wt{\phi}_k - \phi_k\|_{\infty} \nonumber \\
&= \|(\wh{\bm{\theta}}_k - \bm{\theta}_{0k})^\top \bm{B} \|_{\infty} + \|\wt{\phi}_k - \phi_k\|_{\infty} \nonumber \\
&\leq \|\wh{\bm{\theta}}_k - \bm{\theta}_{0k} \|_{\infty} \sup_t \sum_{j = 1}^L |B_j(t)| + \|\wt{\phi}_k - \phi_k\|_{\infty} \nonumber \\
&= \|\wh{\bm{\theta}}_k - \bm{\theta}_{0k} \|_{\infty} + \|\wt{\phi}_k - \phi_k\|_{\infty} \nonumber \\
&\leq \|\wh{\bm{\theta}}_k - \bm{\theta}_{0k} \|_2 + \|\wt{\phi}_k - \phi_k\|_{\infty} \nonumber \\
&= O_p(N^{-1/2} K) + O(K^{-r}) \nonumber \\
&= O_p(N^{-1/2} K). \nonumber
\end{align}
The last equality is obtained from Assumption \ref{Ass_tuning}. The proof is completed.

\end{proof}

\begin{lemma}\label{dominateDer}
Under Assumptions \ref{Ass_phi} - \ref{Ass_tuning}, we have (\ref{dominate}).
\end{lemma}

\begin{proof}[Proof of Lemma \ref{dominateDer}]
We have
\begin{align}
&-\frac{1}{2N} \bm{u}^\top \nabla^2 L(\bm{\Omega}_0, \bm{\xi}_0) \bm{u} \alpha_N^2 \nonumber \\
=& -\frac{1}{2} \bm{u}^\top \Big \{ \frac{1}{N} \nabla^2 L(\bm{\Omega}_0, \bm{\xi}_0) + I(\bm{\Omega}_0, \bm{\xi}_0) \Big \} \bm{u} \alpha_N^2 + \frac{\alpha_N^2}{2} \bm{u}^\top I(\bm{\Omega}_0, \bm{\xi}_0) \bm{u}. \nonumber
\end{align}
According to Chebyshev's inequality, for any $\epsilon > 0$, as $K = o(N^{1/4})$,
\begin{align}
&P \Big \{ \Big \|\frac{1}{N} \nabla^2 L(\bm{\Omega}_0, \bm{\xi}_0) + I(\bm{\Omega}_0, \bm{\xi}_0) \Big \| \geq \frac{\epsilon}{K} \Big \} \nonumber \\
\leq& \frac{K^2}{\epsilon^2} E \Big \{ \Big \|\frac{1}{N} \nabla^2 L(\bm{\Omega}_0, \bm{\xi}_0) + I(\bm{\Omega}_0, \bm{\xi}_0) \Big \|^2 \Big \} \nonumber \\
=& \frac{K^2}{\epsilon^2} E \Big [ \Big \|\frac{1}{N} \nabla^2 L(\bm{\Omega}_0, \bm{\xi}_0) -\frac{1}{N} E \{ \nabla^2 L(\bm{\Omega}_0, \bm{\xi}_0) \} \Big \|^2 \Big ] \nonumber \\
=& \frac{K^2}{N^2 \epsilon^2} E \Big [ \Big \| \nabla^2 L(\bm{\Omega}_0, \bm{\xi}_0) - E \{ \nabla^2 L(\bm{\Omega}_0, \bm{\xi}_0) \} \Big \|^2 \Big ] \nonumber \\
=& \frac{K^2}{N^2 \epsilon^2} \sum_{i, j = 1}^{pK} E \Big ( \frac{\partial^2 L(\bm{\Omega}_0, \bm{\xi}_0)}{\partial \theta_i \partial \theta_j } - E \frac{\partial^2 L(\bm{\Omega}_0, \bm{\xi}_0)}{\partial \theta_i \partial \theta_j } \Big )^2 \nonumber \\
=& \frac{K^2}{N^2 \epsilon^2} \sum_{i, j = 1}^{pK} \mbox{var} \Big ( \frac{\partial^2 L(\bm{\Omega}_0, \bm{\xi}_0)}{\partial \theta_i \partial \theta_j } \Big ) \nonumber \\
\leq& \frac{K^2}{N^2 \epsilon^2} \cdot p^2 K^2 N C_2 \nonumber \\
=& \frac{p^2 C_2 K^4}{N \epsilon^2} \rightarrow 0, \nonumber
\end{align}
where $C_2$ is a constant and $\mbox{var} \Big ( \frac{\partial^2 L(\bm{\Omega}_0, \bm{\xi}_0)}{\partial \theta_i \partial \theta_j } \Big )$ is bounded by $N C_2$.
That means $\Big \|\frac{1}{N} \nabla^2 L(\bm{\Omega}_0, \bm{\xi}_0) + I(\bm{\Omega}_0, \bm{\xi}_0) \Big \| = o_p(K^{-1})$.
Then we have
\begin{align}
-\frac{1}{2N} \bm{u}^\top \nabla^2 L(\bm{\Omega}_0, \bm{\xi}_0) \bm{u} \alpha_N^2 = \frac{\alpha_N^2}{2} \bm{u}^\top I(\bm{\Omega}_0, \bm{\xi}_0) \bm{u} + o_p(\alpha_N^2). \nonumber
\end{align}

\end{proof}

\subsection{Proof of Theorem \ref{theorySpar}}

Define
\begin{align}
\mathcal{T}_k^{(1)} &= \{t \in \mathcal{T}: |\phi_k (t)| > a C (\lambda + K^{-r})\}, \nonumber \\
\mathcal{T}_k^{(2)} &= \{t \in \mathcal{T}: \phi_k (t) = 0\}, \nonumber \\
\mathcal{T}_k^{(3)} &= \mathcal{T} - \mathcal{T}_k^{(1)} - \mathcal{T}_k^{(2)}. \nonumber
\end{align}
We further define $\mathcal{S}_l = \mbox{SUPP}(B_l), l = 1, \ldots, L$. Let $\mathcal{A}_k^{(j)} = \{ l: \mathcal{S}_l \subset \mathcal{T}_k^{(j)} \}, j = 1, 2$ and $\mathcal{A}_k^{(3)} = \{1, \ldots, L\} - \mathcal{A}_k^{(1)} - \mathcal{A}_k^{(2)}$.

\begin{proof}[Proof of Theorem \ref{theorySpar}]
Consider $\theta_{kl}$, where $l \in \mathcal{A}_k^{(2)}$. We have
\begin{align}
\frac{\partial \wt{Q} (\wh{\bm{\Omega}}, \bm{\xi}_0)}{\partial \theta_{kl}} =& - \frac{1}{N} \frac{\partial L(\wh{\bm{\Omega}}, \bm{\xi}_0)}{\partial \theta_{kl}} + \kappa_{\bm{\theta}} \frac{\partial (\bm{\theta}_k^\top V \bm{\theta}_k)}{\partial \theta_{kl}} \Big |_{\bm{\theta}_k = \wh{\bm{\theta}}_k} \nonumber \\
&+ \frac{K + 1}{8T} \int_{\mathcal{T}} p_{\lambda}^{\prime} (|\bm{B}^\top (t) \bm{\theta}_k|) \Big |_{\bm{\theta}_k = \wh{\bm{\theta}}_k} B_l(t) \mbox{sgn} (\wh{\theta}_{kl}) dt \nonumber \\
=& -\frac{1}{N} \Big \{ \frac{\partial L(\bm{\Omega}_0, \bm{\xi}_0)}{\partial \theta_{kl}} + \frac{1}{2} \sum_{g = 1}^L \frac{\partial^2 L(\bm{\Omega}_0, \bm{\xi}_0)}{\partial \theta_{kl} \partial \theta_{kg}} (\wh{\theta}_{kg} - \theta_{0kg}) \Big \} \nonumber \\
&+ \kappa_{\bm{\theta}} \frac{\partial (\bm{\theta}_k^\top V \bm{\theta}_k)}{\partial \theta_{kl}} \Big |_{\bm{\theta}_k = \wh{\bm{\theta}}_k} \nonumber \\
&+ \frac{K + 1}{8T} \mbox{sgn} (\wh{\theta}_{kl}) \int_{\mathcal{S}_l} p_{\lambda}^{\prime} (|\bm{B}^\top (t) \bm{\theta}_k|) \Big |_{\bm{\theta}_k = \wh{\bm{\theta}}_k} B_l(t) dt. \nonumber
\end{align}
Then
\begin{align}
&\Bigg |\lambda^{-1} \frac{\partial \wt{Q} (\wh{\bm{\Omega}}, \bm{\xi}_0)}{\partial \theta_{kl}} - \frac{K + 1}{8T} \mbox{sgn} (\wh{\theta}_{kl}) \int_{\mathcal{S}_l} \lambda^{-1} p_{\lambda}^{\prime} (|\bm{B}^\top (t) \bm{\theta}_k|) \Big |_{\bm{\theta}_k = \wh{\bm{\theta}}_k} B_l(t) dt \Bigg | \nonumber \\
=& \Bigg | - \lambda^{-1} \frac{1}{N} \Big \{ \frac{\partial L(\bm{\Omega}_0, \bm{\xi}_0)}{\partial \theta_{kl}} + \frac{1}{2} \sum_{g = 1}^L \frac{\partial^2 L(\bm{\Omega}_0, \bm{\xi}_0)}{\partial \theta_{kl} \partial \theta_{kg}} (\wh{\theta}_{kg} - \theta_{0kg}) \Big \} \nonumber \\
& \qquad \qquad \qquad \qquad \qquad \qquad \qquad \qquad \qquad \qquad + \lambda^{-1} \kappa_{\bm{\theta}} \frac{\partial (\bm{\theta}_k^\top V \bm{\theta}_k)}{\partial \theta_{kl}} \Big |_{\bm{\theta}_k = \wh{\bm{\theta}}_k} \Bigg | \nonumber \\
\leq& O_p(N^{-1/2} K^{3/2} \lambda^{-1}) + o_p(N^{-1/2} \lambda^{-1}) \rightarrow 0. \label{sign}
\end{align}
As
\begin{align}
\mathop{\lim \inf}_{N \rightarrow \infty} \mathop{\lim \inf}_{x \rightarrow 0^{+}} \lambda^{-1} p_{\lambda}^{\prime}(x) > 0, \nonumber
\end{align}
the sign of $\frac{K + 1}{8T} \mbox{sgn} (\wh{\theta}_{kl}) \int_{\mathcal{S}_l} \lambda^{-1} p_{\lambda}^{\prime} (|\bm{B}^\top (t) \bm{\theta}_k|) \Big |_{\bm{\theta}_k = \wh{\bm{\theta}}_k} B_l(t) dt$ is determined by $\wh{\theta}_{kl}$. Hence, the sign of $\frac{\partial \wt{Q} (\wh{\bm{\Omega}}, \bm{\xi}_0)}{\partial \theta_{kl}}$ is determined by $\wh{\theta}_{kl}$ according to (\ref{sign}). Since $\wh{\theta}_{kl}$ is the local minimizer of $\wt{Q} (\bm{\Omega}, \bm{\xi}_0)$, we have $\frac{\partial \wt{Q} (\wh{\bm{\Omega}}, \bm{\xi}_0)}{\partial \theta_{kl}} = 0$, thus $\wh{\theta}_{kl} = 0$. That means $\wh{\theta}_{kl} = 0$ for all $l \in \mathcal{A}_k^{(2)}$ in probability.

Define $\wh{\mathcal{A}}_k^{(2)} = \{l \in \mathcal{A}_k^{(2)}: \wh{\theta}_{kl} = 0\}$. We have $\wh{\mathcal{A}}_k^{(2)} = \mathcal{A}_k^{(2)}$ in probability. Moreover, $\bigcup_{l \in \wh{\mathcal{A}}_k^{(2)}} \mathcal{S}_l = \bigcup_{l \in \mathcal{A}_k^{(2)}} \mathcal{S}_l$ in probability. Since $\bigcup_{l \in \mathcal{A}_k^{(2)}} \mathcal{S}_l$ converges to $\mbox{NULL}(\phi_k)$ as $K \rightarrow \infty$ according to the compact support property of B-spline basis, we have \begin{align}\label{union_con}
\bigcup_{l \in \wh{\mathcal{A}}_k^{(2)}} \mathcal{S}_l \rightarrow \mbox{NULL}(\phi_k)
\end{align}
in probability.

We further want to show $\mathcal{T}_k^{(1)} \subset \mbox{SUPP}(\wh{\phi}_k)$ in probability. By Theorem \ref{theoryconsis}, $\|\wh{\phi}_k - \phi_k\|_{\infty} = O_p(N^{-1/2} K + K^{-r}) = O_p(\lambda + K^{-r})$. Thus for any $\epsilon > 0$, there exists some constant $C_4 > 0$ such that $P \{|\wh{\phi}_k(t) - \phi_k(t)| \leq C_4 a (\lambda + K^{-r}), t \in \mathcal{T}_k^{(1)} \} > 1 - \epsilon$. Let $C = 2C_4$ and making use of the definition of $\mathcal{T}_k^{(1)}$, we have
\begin{align}
P\{|\wh{\phi}_k(t)| \geq C_4 a (\lambda + K^{-r}), t \in \mathcal{T}_k^{(1)} \} > 1 - \epsilon. \nonumber
\end{align}
Since $C_4 a (\lambda + K^{-r}) > 0$, we have $\mathcal{T}_k^{(1)} \subset \mbox{SUPP}(\wh{\phi}_k)$ in probability. Thus $\mbox{NULL}(\wh{\phi}_k) \subset \mathcal{T}_k^{(2)} \cup \mathcal{T}_k^{(3)}$. Further, as $N \rightarrow \infty$ and $K \rightarrow \infty$, we have
\begin{align}\label{subset_relation}
\bigcup_{l \in \wh{\mathcal{A}}_k^{(2)}} \mathcal{S}_l \subset \mbox{NULL}(\wh{\phi}_k) \subset \mathcal{T}_k^{(2)} \cup \mathcal{T}_k^{(3)} = \mbox{NULL}(\phi_k) \cup \mathcal{T}_k^{(3)}.
\end{align}
By (\ref{union_con}), (\ref{subset_relation}) and the fact that $\mathcal{T}_k^{(3)}$ converges to $\emptyset$, we have $\mbox{NULL}(\wh{\phi}_k) \rightarrow \mbox{NULL}(\phi_k)$ and $\mbox{SUPP}(\wh{\phi}_k) \rightarrow \mbox{SUPP}(\phi_k)$ in probability. The proof is completed.

\end{proof}

\subsection{Proof of Theorem \ref{theoryXi}}

\begin{proof}[Proof of Theorem \ref{theoryXi}]
For $\bm{\xi}_i$, the objective function is
\begin{align}
Q^{\ast}(\bm{\Theta}_0, \bm{\xi}_i) = - \sum_{j = 1}^{m_i} \log \pi \{q_{ij} (\textbf{B}_{ij}^\top \bm{\Theta}_0^\top \bm{\xi}_i) \}. \nonumber
\end{align}
Let $\beta_M = M^{-1/2}$. We want to show that for any $\epsilon > 0$, $\exists C_3 > 0$, $\forall D > C_3$, such that
\begin{align}\label{theorygoalxi}
P \Big \{ \inf_{\|\bm{u}\|_2 = D} Q^{\ast}(\bm{\Theta}_0, \bm{\xi}_{0i} + \beta_M \bm{u}) > Q^{\ast}(\bm{\Theta}_0, \bm{\xi}_{0i}) \Big \} \geq 1 - \epsilon,
\end{align}
where $\bm{\xi}_{0i} = (\xi_{i1}, \ldots, \xi_{ip})^\top$ is the true parameter. It indicates there exists a local minimizer in the ball $\{\bm{\xi}_{0i} + \beta_M \bm{u}: \|\bm{u}\|_2 \leq D\}$, with probability at least $1 - \epsilon$. Moreover, the local minimizer satisfies $\|\wh{\bm{\xi}}_{i} - \bm{\xi}_{0i}\|_2 = O_p (\beta_M)$.

Similarly, define $D_M(\bm{u}) = Q^{\ast}(\bm{\Theta}_0, \bm{\xi}_{0i} + \beta_M \bm{u}) - Q^{\ast}(\bm{\Theta}_0, \bm{\xi}_{0i})$. Then
\begin{align}
D_M(\bm{u}) &= - \Big [ \sum_{j = 1}^{m_i} \log \pi [q_{ij} \{\textbf{B}_{ij}^\top \bm{\Theta}_0^\top (\bm{\xi}_{0i} + \beta_M \bm{u})\} ] - \sum_{j = 1}^{m_i} \log \pi \{q_{ij} (\textbf{B}_{ij}^\top \bm{\Theta}_0^\top \bm{\xi}_{0i}) \} \Big ] \nonumber \\
&\triangleq - \{ l(\bm{\Theta}_0, \bm{\xi}_{0i} + \beta_M \bm{u}) - l(\bm{\Theta}_0, \bm{\xi}_{0i}) \} \nonumber \\
&= - \beta_M \nabla^\top l (\bm{\Theta}_0, \bm{\xi}_{0i}) \bm{u} + \frac{M \beta_M^2}{2} \bm{u}^\top I(\bm{\Theta}_0, \bm{\xi}_{0i}) \bm{u} \{1 + o_p(1)\} \nonumber \\
&\triangleq \delta_1 + \delta_2. \nonumber
\end{align}
In specific, as
\begin{align}
|\delta_1| = |\beta_M \nabla^\top l (\bm{\Theta}_0, \bm{\xi}_{0i}) \bm{u}| \leq \beta_M \|\nabla^\top l (\bm{\Theta}_0, \bm{\xi}_{0i})\|_2 \|\bm{u}\|_2 = O_p(M^{1/2} \beta_M), \nonumber
\end{align}
$D_M(\bm{u})$ is dominated by $\delta_2$ with a sufficient large $D$. That means we have (\ref{theorygoalxi}) using a sufficient large $D$. Since $\|\wh{\bm{\xi}}_{i} - \bm{\xi}_{0i}\|_2 = O_p (\beta_M)$, $|\wh{\xi}_{ik} - \xi_{ik}| = O_p(\beta_M)$ for all $k = 1, \ldots, p$. The proof is completed.

\end{proof}

\bibliographystyle{unsrtnat}
\bibliography{ref}

\end{document}